\documentclass[12pt]{article}

\usepackage[margin=1.25in]{geometry}
\usepackage{graphicx}
\usepackage{amsthm}
\usepackage{amsfonts}
\usepackage{amssymb}
\usepackage{bm}
\usepackage{amsmath}
\usepackage{color}
\usepackage{microtype}
\usepackage{enumitem}
\usepackage{xcolor}
\usepackage{listings}
\usepackage[toc,page]{appendix}
\usepackage{relsize}
\usepackage{lineno}
\usepackage{multirow}
\usepackage{verbatim}
\usepackage{changepage}
\usepackage{url}
\newtheorem{theorem}{Theorem}
\newtheorem{assumption}{Assumption}

\newtheorem{corollary}{Corollary}

\newtheorem{lemma}{Lemma}
\newtheorem{property}{Property}
\newtheorem{proposition}{Proposition}[subsection]

\newtheorem{definition}{Definition}
\newtheorem{example}{Example}
\newtheorem{remark}{Remark}

\theoremstyle{definition}
\graphicspath{ {FigsEtc/} }
\newtheorem*{excont}{Example \continuation}
\newcommand{\continuation}{??}
\newenvironment{continueexample}[1]
 {\renewcommand{\continuation}{\ref{#1}}\excont[continued]}
 {\endexcont}

\begin{document}

\begin{center}
{\LARGE Monotone Comparative Statics

for Equilibrium Problems}

\bigskip
\bigskip
\bigskip
\bigskip

Alfred Galichon$^{\dag }$, Larry Samuelson{\small $^{\flat }$}, and Lucas Vernet{\small $^{\S }$}

\bigskip
\bigskip

\today

\bigskip
\bigskip
\bigskip

\end{center}

\noindent {\bf Abstract.~~}  
We introduce a notion of  substitutability for correspondences and establish a monotone comparative static result,  unifying results such as the inverse isotonicity of M-matrices, Berry, Gandhi and Haile's identification of demand systems, monotone comparative statics, and results on the structure of the core of matching games without transfers (Gale and Shapley) and with transfers (Demange and Gale).  More specifically, we introduce the notions of {\em unified gross substitutes} and {\em nonreversingness} and show  that if  $\mathtt Q:P\rightrightarrows Q$ is a supply correspondence defined on a set of prices $P$ which is a sublattice of $\mathbb{R}^N$, and $\mathtt Q$ satisfies these two properties, then the set of equilibrium prices $\mathtt Q^{-1}(q)$ associated with a vector of quantities $q\in Q$ is increasing (in the  strong set order) in  $q$; and it is a sublattice of $P$.  

\vfill

$^{\dag }$Economics Department, FAS, and Mathematics Department, Courant Institute, New York University; and Economics Department, Sciences Po. Email: ag133@nyu.edu.  Galichon gratefully acknowledges funding from NSF grant DMS-1716489 and ERC grant CoG-866274.

{\small $^{\flat }$}Department of Economics, Yale University. Email: larry.samuelson@yale.edu.

{\small $^{\S }$}Economics Department, Sciences Po, and Banque de France. Email: 

\noindent lucas.vernet@acpr.banque-france.fr.

\thispagestyle{empty}
\clearpage

$~$
\vspace{-.75in}

{\footnotesize

\tableofcontents}

\thispagestyle{empty}

\clearpage


\setcounter{page}{1}

\begin{center}
{\Large Monotone Comparative Statics for Equilibrium Problems}
\end{center}

\section{Introduction}

This paper proposes the notion of {\em unified gross substitutes} for a correspondence $\mathtt Q:P\rightrightarrows Q$.  For concreteness we often interpret $\mathtt Q$ as a supply correspondence mapping from a set of prices $P$ to a set of quantities $Q$, though our analysis applies to correspondences in general.  Our analysis encompasses the familiar case in which  $\mathtt Q$ arises from the optimization problem of a single agent, but  unified gross substitutes need not  refer to a single agent's decision problem, and we are especially interested in its potential for the study of equilibrium problems.

Our focus is the inverse isotonicity of the correspondence $\mathtt{Q}$. We show that if the correspondence $\mathtt Q$ satisfies unified gross substitutes as well as a mild condition called \emph{nonreversingness}, then the set of parameters $\mathtt Q^{-1}(q)$ associated with an element $q\in Q$ is increasing (in the strong set order) in  $q$ and is a sublattice of $P$.

For functions, the notion of unified gross substitutes is equivalent to the familiar notion of weak gross substitutes.   Berry, Gandhi and Haile \cite{berry2013connected} provide an inverse isotonicity result for functions and explain the importance of inverse isotonicity.  One can view our work as generalizing their result to correspondences.   In some cases, there is a natural formulation of the inverse $\mathtt Q^{-1}$ as the solution to an optimization problem.  We can then (under appropriate conditions) apply the monotone comparative statics results of Topkis \cite{Topkis1998} and Milgrom and Shannon \cite{MandS94} to obtain inverse isotonicity.  One can view our work as extending the study of monotone comparative statics beyond optimization problems. 

For correspondences, the notion of unified gross substitutes implies (but is not equivalent to) Kelso and Crawford's \cite{KandC1982} notion of gross substitutes for correspondences, which is too weak to imply our inverse isotonicity result.   The notion of unified gross substitutes is related to the generalization of $M$-matrices to $M$-functions introduced by More and Rheinboldt \cite{more1973p} and is independent of a similar substitutes notion introduced by  Polterovich and Spivak \cite{Posp1983}.

We show that unified gross substitutes for the argmax correspondence of a maximization problem is equivalent to the submodularity of the value function, generalizing a familiar result for argmax functions.  We introduce a new equilibrium problem, referred to as the equilibrium flow problem, that contains a number of familiar settings as special cases and whose equilibrium correspondence satisfies unified gross substitutes.  This provides a new route to the result that the set of stable matches in matching problems with transfers (as in Demange and Gale \cite{demange1985strategy}) and without transfers form a sublattice.  Finally, we examine  hedonic pricing problems (cf. Rosen \cite{Rosen1974} and Ekeland, Heckman and Nesheim \cite{EHN2004}), extending the basic results of  Chiappori, McCann and Nesheim \cite{CMandN2010} 
beyond quasilinear utilities and establishing an inverse isotonicity result for such models.

\section{Theory}

When $u$ and $v$ are two elements of a partially ordered set $(\mathcal{S},\leq)$, we use $u < v$ to express that $u\leq v$ and $u \neq v$. Let $P\subseteq \mathbb{R}^{N}$ and $Q \subseteq \mathbb R^N$, for some finite $N$, with  generic elements  $p\in P$ and $q\in Q$. Let $\mathtt  Q:P\rightrightarrows Q$ be a correspondence.  We maintain the following assumption throughout, typically without explicit mention:
\begin{assumption}\label{swat}
$P$ is a sublattice of $\mathbb{R}^N$.
\end{assumption}

\noindent Recall from Topkis \cite[p. 13]{Topkis1998} that a set $P$ is a \emph{sublattice of $\mathbb{R}^N$} if for any pair of vectors $p,p'\in P$, the set $P$ also contains  their least upper bound (denoted by $p\vee p'$ and defined as the coordinate-wise maximum of $p$ and $p'$)  as well as their  greatest lower bound (denoted by  $p\wedge p'$ and defined as the  coordinate-wise minimum of $p$ and $p'$).

In one of our leading interpretations of the correspondence $\mathtt Q$, we view the dimensions of $\mathbb R^N$ as identifying goods and interpret $\mathtt Q$ as a supply correspondence. An element $p\in P$ is then a price vector, with $p_{z}$ denoting the price of good $z\in\{1,\ldots,N\}$. An element $q\in \mathtt Q(p)$ is an allocation, with $q_{z}$ denoting the quantity of good $z$ supplied at price vector $p$.  No matter what the interpretation, we typically refer to elements of $q$ as quantities and elements of $p$ as prices.

\subsection{Unified Gross Substitutes}

\subsubsection{Definition}
Our basic notion of substitutability for correspondences is:%
%
%

\begin{definition}[{\bf Unified Gross Substitutes}]\label{duck}
The correspondence $\mathtt Q:P\rightrightarrows Q$ satisfies \emph{unified gross substitutes} if, given $p\in P, p' \in P$, $q\in \mathtt Q\left( p\right) $ and $q^{\prime }\in \mathtt Q\left( p^{\prime }\right) $, there exists $q^{\wedge }\in \mathtt Q\left( p\wedge p^{\prime }\right) $ and $q^{\vee }\in \mathtt Q\left( p\vee p^{\prime }\right) $ such that
\begin{eqnarray}
p_{z}\leq p_{z}^{\prime }&\implies& q_{z}\leq q_{z}^{\wedge }\text{ and }%
q_{z}^{\vee }\leq q_{z}^{\prime }~ \label{lester}\\
p_{z}^{\prime }<p_{z}&\implies& q_{z}^{\prime }\leq q_{z}^{\wedge }\text{ and }%
q_{z}^{\vee }\leq q_{z}.\label{flatt}
\end{eqnarray}
\end{definition}

\begin{figure}[t]
    \centering
    \includegraphics[width=0.7\textwidth]{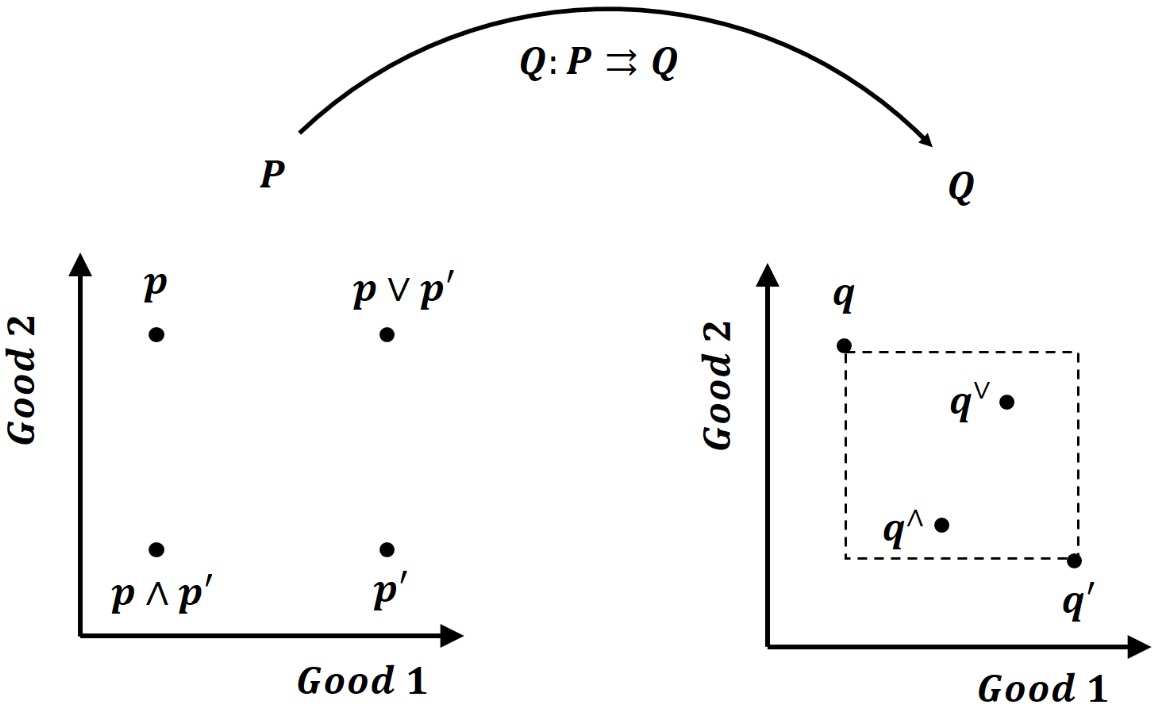}
    \caption{An illustration of the unified gross substitutes property with two goods.  $\mathtt Q$ is a mapping between the set of prices on the left and the set of quantities on the right with $q\in \mathtt Q(p)$ and $q'\in \mathtt Q(p')$. Definition~\ref{duck} imposes that there exists a $q^{\vee }\in \mathtt (p\vee p')$, smaller than the top-right border of the rectangle drawn in the set of quantities, and a $q^{\wedge}\in \mathtt (p\wedge p')$ larger than the bottom-left border.}
    \label{fig:UGS_illustration}
\end{figure}

\noindent Figure~\ref{fig:UGS_illustration} gives an illustration of the unified gross substitutes property with two goods. This definition is appropriate for our interpretation of $\mathtt Q$ as a supply correspondence, and would need to be adjusted in a straightforward way for applications to demand correspondences.%
\footnote{The counterpart of \eqref{lester}--\eqref{flatt} would then be
\begin{eqnarray}
	p_{z}\leq p_{z}^{\prime }&\implies& q_{z}\geq q_{z}^{\wedge}\text{ and }%
	q_{z}^{\vee}\geq q_{z}^{\prime}\label{white}\\
	p_{z}^{\prime }<p_{z}&\implies& q_{z}^{\prime}\geq q_{z}^{\wedge}\text{ and }%
	q_{z}^{\vee}\geq q_{z}.\label{rabbit}
\end{eqnarray}
}
 Section~\ref{sect:connections} connects this definition with existing notions in the literature and explains what is unified by this concept. 

One might be tempted to strengthen this definition by making the antecedents of \eqref{lester}--\eqref{flatt} both weak inequalities.  The resulting notion is not only stronger than we need, but vitiates many of our equivalence results and too often fails to hold. Appendix~{\ref{par:weak-strong}} provides an example.  Notice that, given the ability to interchange the prices $p$ and $p'$, it is irrelevant which antecedent carries the strict inequality, as long as one does. 

\subsubsection{Properties}\label{properties-ugs}

This section collects some basic properties of unified gross substitutes.  We first show that in the case of a function, unified gross substitutes is equivalent to the the textbook notion of weak gross substitutes (e.g., Mas-Colell, Green and Whinston \cite[Definition 17.F.2, p. 611]{MWG1995}).%
\footnote{Requiring $\mathtt{q}_{i}\left( p\right) $ to be strictly decreasing in $p_{j}$ gives {\em strict} gross substitutes.  Mas-Colell, Green and Whinston \cite{MWG1995} provide the definition of weak gross substitutes  for a demand function, namely that $\mathtt{q}_{i}\left( p\right) $ is nondecreasing in $p_{j}$.  Traditionally, goods are said to be simply ``substitutes'' when referring to Hicksian demand functions and ``gross substitutes'' when referring to Walrasian demand functions.  We do not restrict to correspondences derived from an optimization problem, obviating this distinction, and typically use the  phrase gross substitutes.}

\begin{property}{\em \textbf{For functions, unified gross substitutes and weak gross substitutes are equivalent.}\label{prop:wgs-implies-ugs} 
Recall that a function $\mathtt{q}:\mathbb{R}^{N}\rightarrow \mathbb{R}^{N}$ satisfies \emph{weak gross substitutes}, or equivalently is \emph{off-diagonally isotone}, or is a \emph{Z-map} (Plemmons \cite{Plemmons1977}), if and only if $\mathtt{q}_{i}\left( p\right) $ is nonincreasing in $p_{j}$
for $i\neq j$. 

We now show that the function $\mathtt q$ satisfies weak gross substitutes if and only if the  correspondence $\mathtt{Q}\left( p\right) =\left\{ \mathtt{q}\left( p\right) \right\} $ satisfies
unified gross substitutes. 
First let $\mathtt q$ satisfy weak gross substitutes, take $q=\mathtt{q}\left( p\right) $ and $q^{\prime
}=\mathtt{q}\left( p^{\prime }\right) $, and let $q^{\wedge }=\mathtt{q}\left( p\wedge
p^{\prime }\right) $ and $q^{\vee }=\mathtt{q}\left( p\vee p^{\prime }\right) $. Then 
$p_{z}\leq p_{z}^{\prime }\implies \left( p\wedge p^{\prime }\right)
_{z}=p_{z}$, which combines with $p\wedge p^{\prime }\leq p$ and weak gross substitutes to give $ \mathtt{q}_{z}\left( p\right) \le \mathtt{q}_{z}\left( p\wedge
p^{\prime }\right)$ and hence the first part of \eqref{lester}. The other requirements in \eqref{lester}--\eqref{flatt}  are dealt with similarly, and hence $\mathtt Q$ satisfies unified gross substitutes.  Conversely, let $\mathtt{Q}\left( p\right) =\left\{ \mathtt{q}\left( p\right) \right\} $ satisfy unified gross substitutes, let $p_j\ge p_j'$ and $p_i = p_i'$ for all $i\neq j$.  Then $p^{\wedge}=p'$ and $p^{\vee}=p$, so that applying \eqref{lester} to some $i\neq j$ gives $\mathtt q_i(p)\le \mathtt q_i(p')$, and hence $\mathtt q$ satisfies weak gross substitutes.
\hfill\rule{.1in}{.1in}
}
\end{property}

We next note that if the correspondence $\mathtt Q$ is the solution to a maximization problem, the unified gross substitutes condition admits a familiar characterization.  Ausubel and Milgrom \cite[Theorem 10]{AandM2002} show that a demand correspondence satisfies a  weak gross substitutes property (namely that $\mathtt{q}_{i}\left( p\right) $ is nonincreasing in $p_{j}$ for $i\neq j$ and for those $p$ for which $\mathtt q(p)$ is single-valued) if and only if the associated indirect utility function is submodular.  We establish an analogous condition for unified gross substitutes.  

\begin{property}\emph{\textbf{Subdifferentials of convex submodular functions satisfy unified gross substitutes.}\label{prop:submodular-implies-ugs} 
Let $c^*:\mathbb{R}^N \to \mathbb{R}$ 
be a convex function and consider the subdifferential of $c^*$, given by  $
 \partial c^*(p) = \left\{ q:p^{\prime }\rightarrow
[q^{\top }p^{\prime }-c^{\ast }\left( p^{\prime }\right)] \text{ is maximal at 
}p\right \}$.   
Theorem \ref{spice} in Section \ref{king} shows that $c^*$ is submodular if and only if the correspondence $p \rightrightarrows \partial c^*(p)$ satisfies unified gross substitutes.}

{\em Section \ref{king} shows that we can think of $c^*$ as the indirect profit function associated with a profit maximization problem, which is necessarily convex, allowing us to conclude that the associated supply correspondence satisfies unified gross substitutes if and only if the indirect profit function is submodular.  
\hfill\rule{.1in}{.1in}
}
\end{property}


Property \ref{pty:walras-law} below illustrates one use of the following property.

\begin{property}\emph{\textbf{Monetary measurement preserves unified gross substitutes.}\label{pty:change-of-num} 
Given an excess supply correspondence $\mathtt Q$ for a competitive economy, we can define the excess supply correspondence expressed in monetary terms as $$\mathtt Q^{\$}\left(p\right)=\left\{ \left(p_{z}q_{z}\right)_{z\in\left\{ 1,...,N\right\} }\mbox{ with }q\in \mathtt Q\left(p\right)\right\}. $$
It is straightforward to verify that when the domain $P$ of $\mathtt Q$ is included in $\mathbb{R}_{+}^{N}$, so that prices are nonnegative, then  $\mathtt Q$ satisfies unified gross substitutes if and only if $\mathtt Q^{\$}$ satisfies unified gross substitutes.  
\hfill\rule{.1in}{.1in}
}
\end{property}

The following property shows that unified gross substitutes aggregates in a natural way. 

\begin{property}\label{bigdog}\emph{\textbf{Aggregation preserves unified gross substitutes.}\label{pty:agg-pres-ugs}
If the correspondences $\mathtt Q^{1}$ and $\mathtt Q^{2}$ have the unified gross substitutes property, then
so does $\lambda \mathtt Q^{1}+\mu \mathtt Q^{2}$ for $\lambda\ge 0$ and $\mu\geq0$, where
\begin{equation*}
\left(\lambda \mathtt Q^{1}+\mu \mathtt Q^{2}\right)\left(p\right)=\left\{ \lambda q^{1}+\mu
q^{2}:q^{1}\in \mathtt Q^{1}\left(p\right),q^{2}\in \mathtt Q^{2}\left(p\right)\right\}.
\end{equation*}
The proof, in Appendix \ref{woodpecker}, is a straightforward bookkeeping exercise.
\hfill\rule{.1in}{.1in}
}
\end{property}

\subsection{Nonreversingness}

Our second condition is a requirement that the correspondence $\mathtt Q$ cannot completely reverse the order of two points:

\subsubsection{Definitions}

\begin{definition}[{\bf Nonreversing correspondence}]
The correspondence $\mathtt Q:P\rightrightarrows Q$ is {\em nonreversing} if 
\begin{equation}\label{hambone}
\left(
\begin{array}{c}
	q\in \mathtt Q\left( p\right)\\
	{	q'}\in \mathtt Q\left( p'\right)\\
	q\leq {q'}\\
	p\geq {p'}\\
\end{array}	
\right)	
\implies 
\left(
\begin{array}{c}
	q\in \mathtt Q\left({p'}\right)\\
	{q'}\in \mathtt Q\left( p\right) 
\end{array}
\right).
\end{equation}
\end{definition}

\noindent Nonreversingness is a weak monotonicity requirement.  It is implied by the (stronger but common) assumption that $\mathtt Q$ is increasing in the strong set order, as we would typically expect of a supply correspondence.

When $\mathtt Q$ is point-valued, the previous notion equivalently boils down to the following implication:
\begin{equation}
p \geq p^\prime\text{ and }   \mathtt q(p) \leq \mathtt q(p^\prime)\text{ implies } \mathtt q(p)=\mathtt q(p^\prime).
\end{equation}

Nonreversingneess is weaker than \emph{strong nonreversingness}, where the conclusion in \eqref{hambone} is replaced by the stronger statement that $p=p^\prime$:

\begin{definition}[{\bf Strongly nonreversing correspondence}]\label{def:strong-nonreversing}
The correspondence $\mathtt Q:P\rightrightarrows Q$ is {\em strongly nonreversing} if

\begin{equation}\label{crepes}
\left(
\begin{array}{c}
	q\in \mathtt Q\left( p\right)\\
	{	q'}\in \mathtt Q\left( p'\right)\\
	q\leq {q'}\\
	p\geq {p'}\\
\end{array}	
\right)	
\implies 
 p = p^\prime.
\end{equation}
\end{definition}

\subsubsection{Properties}

Some familiar properties of
economic models imply the nonreversingness of the supply correspondence.   The first three of the following properties immediately imply nonreversingness by ensuring that the antecedent of \eqref{hambone} can hold only if $q=q'$, rendering the consequent immediate.  

\begin{property}\emph{\textbf{Constant aggregate output implies nonreversingness.} \label{pty:cao-implies-nr}
The correspondence $\mathtt{Q}$ satisfies \emph{constant aggregate output} if there exists $k \in \mathbb{R}^N_{++}$ such that $\sum_{z=1}^N k_z q_z=0$ holds for all  $p \in P$ and $q\in \mathtt{Q}\left(p\right)$.}
	
{\em   It will often be natural to take $k$ to be the unit vector.  One obvious circumstance in which aggregate output is constant is that in which the bundle of goods is augmented by an ``outside good'' whose quantity is the negative of the sum of the quantities of the original set of goods, as in Berry, Gandhi and Haile \cite{berry2013connected}.  Alternatively, the correspondence may be describing market shares in a model of competition, probabilities in a prediction problem, or budget shares in a model of consumption.
\hfill\rule{.1in}{.1in}		
}		
\end{property}

\begin{property}\emph{\textbf{Monotone total output implies nonreversingness.} \label{pty:mto-implies-nr}
The correspondence $\mathtt{Q}$ satisfies \emph{monotone total output} 
if for $q \in \mathtt{Q}(p)$ and $q^\prime \in \mathtt{Q}(p^\prime)$, $p \geq p^\prime$ implies $\sum_{z=1}^N q_z \geq \sum_{z=1}^N q^\prime_z$.   } 

{\em  The monotone total output property plays an important role in the matching with contracts literature (Hatfield and Milgrom \cite{HandM2005}) under the name of \emph{law of aggregate demand} (or more properly here, law of aggregate supply).%
\footnote{That literature assumes goods are idiosyncratic and indivisible, so $q \in \{0,1\}^N$ and   $\sum_{z=1}^N q_z $ is the cardinality of the supply bundle $q$.}
\hfill\rule{.1in}{.1in}
}
\end{property}

\begin{property}\emph{\textbf{Aggregate monotonicity implies nonreversingness.}\label{zeus} The correspondence  $\mathtt{Q}$ satisfies \emph{aggregate monotonicity} if for $q \in \mathtt{Q}(p)$ and $q^\prime \in \mathtt{Q}(p^\prime)$, both  $p \geq p^\prime$ and $q < q^\prime $ cannot hold simultaneously.  }
	
{\em The correspondence $\mathtt Q$ satisfies aggregate monotonicity if and only if
for all $p\ge p'$, $q\in \mathtt Q(p)$ and $q'\in \mathtt Q(p')$, there exists $k\in \mathbb R_{++}$ with $\sum_{z=1}^N k_z q_z \ge \sum_{z=1}^N k_z q'_z$.
%
}
\hfill\rule{.1in}{.1in}

\end{property}

Constant aggregate output implies monotone total output if we take $k$ to be the unit vector in the former; otherwise they are not nested.  Aggregate monotonicity is clearly implied by  either of constant aggregate output or  monotone total output, but the converses fail.  
To illustrate, Appendix \ref{kettle} provides an example of a function that is nonreversing and that satisfies aggregate and weighted monotonicity (defined next) but not monotone total output.  

Our next property combines with unified gross substitutes to imply nonreversingness.  

\begin{property}\emph{\textbf{Weighted monotonicity and unified gross substitutes imply nonreversingness.}
The correspondence $\mathtt Q$ satisfies {\em weighted monotonicity} if, given prices $p$ and $p^{\prime }$ in $P$ and allocations $q\in \mathtt Q\left( p\right) $ and $q^{\prime 	}\in \mathtt Q\left( p^{\prime }\right) $, there exists  $k\in\mathbb R^N_{++}$ and allocations  $q^{\wedge }\in \mathtt Q\left(	p\wedge p^{\prime }\right) $ and $q^{\vee }\in \mathtt Q\left( p\vee p^{\prime	}\right) $ such that
\begin{equation}\label{bug}
\sum_{z=1}^Nk_zq_{z}\geq \sum_{z=1}^Nk_zq_{z}^{\wedge }~~~~{\rm
and }~~~~\sum_{z=1}^Nk_zq_{z}^{\vee }\geq \sum_{z=1}^Nk_zq_{z}^{\prime }.
\end{equation}
Notice that the vector $k$ is allowed to depend on the prices and allocations involved. }

{\em Weighted monotonicity and unified gross substitutes together imply that  $\mathtt Q$ is nonreversing.  To see this, consider $p$ and $p^{\prime }$ in $P$, and $q$ and $q^{\prime }$ in $Q$, such that $ q\in \mathtt Q\left( p\right) $, $q^{\prime }\in \mathtt Q\left( p^{\prime }\right) $, $q\leq q^{\prime }$ and $p\ge p'$.  By definition, there
exists $k\in \mathbb R^N_{++}$, $q^{\wedge }\in \mathtt Q\left( p\wedge p^{\prime }\right) $ and $q^{\vee
}\in \mathtt Q\left( p\vee p^{\prime }\right) $, satisfying \eqref{bug}.
Using $q\le q'$ and unified gross substitutes \eqref{lester}, one sees that $q_{z}\leq q_{z}^{\wedge }$ and $q_{z}^{\vee }\leq
q_{z}^{\prime }$\ for all $z$. By summation one obtains $%
\sum_{z=1}^Nk_zq_{z}\leq \sum_{z=1}^Nk_zq_{z}^{\wedge }$ and
$\sum_{z=1}^Nk_zq_{z}^{\vee }\leq \sum_{z=1}^Nk_zq_{z}^{\prime }$, and thus (using \eqref{bug}) all these inequalities are equalities. Hence $%
q=q^{\wedge }$, and $q^{\prime }=q^{\vee }$, which shows that $q\in \mathtt Q\left(
p\wedge p^{\prime }\right)=\mathtt Q(p') $ and $q^{\prime }\in \mathtt Q\left( p\vee p^{\prime
}\right) =\mathtt Q(p)$ as needed.
\hfill\rule{.1in}{.1in}
}
\end{property}

\begin{property}\emph{\textbf{Walras law and nonreversingness.}\label{pty:walras-law} 
Given a correspondence $\mathtt Q$, consider the correspondence measured in monetary terms,
\[
\mathtt Q^{\$}\left(p\right)=\left\{ \left(p_{z}q_{z}\right)_{z\in\left\{ 1,...,N\right\} }\mbox{ with }q\in \mathtt Q\left(p\right)\right\}
\]
introduced in Property \ref{pty:change-of-num}.  Then $\mathtt Q$ satisfies Walras law if for all $ p$, and all $q\in \mathtt Q\left(p\right)$,we have  $p^{\mathrm{\top}}q=0$, which holds
if and only if for all $ p$ and all  $\tilde q\in \mathtt Q^{\$}\left(p\right)$,
 we have $1^{\mathrm{\top}}\tilde q=0$. Hence, if $\mathtt Q$ satisfies Walras law, it follows that $\mathtt Q^{\$}$ has constant aggregate output, and is therefore nonreversing.}

{\em  The obvious example of a correspondence satisfying Walras law is the excess supply correspondence of a competitive economy.
\hfill\rule{.1in}{.1in}
}
\end{property}

\begin{property}\emph{\textbf{Inverse isotonicity implies nonreversingness.}\label{monica} 
The inverse correspondence $\mathtt{Q}^{-1}$ is isotone in the strong set order  if for $q\in\mathtt{Q}(p)$ and $q^\prime  \in\mathtt{Q}(p^\prime )$ with  $q \leq q^\prime$, one has  
$q\in\mathtt{Q}(p \wedge p^ \prime)$ and $q^\prime \in\mathtt{Q}(p \vee p^ \prime)$. }

{\em It follows immediately from the definitions that inverse isotonicity in the strong set order implies nonreversingness.  Section \ref{grenade} develops this connection further.
\hfill\rule{.1in}{.1in}
}
\end{property}


\begin{property}\emph{
\textbf{Single crossing of the objective function implies nonreversing argmax, and conversely.}\label{indiff} Consider $\varphi:\mathbb R^{2N}\to \mathbb R$ and assume  $\varphi(p,q)$ has \emph{single crossing} in $(p,q)$, that is if $0 \geq   \varphi(p',q) -\varphi(p',q')  $ and  $\varphi(p,q) - \varphi(p,q')\geq 0$ holds for $p' \geq p$ and $q' \leq q$, then these two inequalities hold as equalities. Although stated differently, this definition is equivalent to the one introduced by Milgrom and Shannon \cite{MandS94}.
Consider 
\begin{equation}\label{Q-argmax}
    \mathtt Q(p)=\arg \max_{q \in Q} \{ \varphi(p,q) \}.
\end{equation}
Then $\mathtt Q$ is nonreversing.
Indeed, assume $q \in \mathtt Q(p)$, $q' \in \mathtt Q(p')$, $p \geq p'$ and $q \leq q'$. Then by definition of $\mathtt Q$, one has $\varphi(p,q)\geq \varphi(p,q')$ and $\varphi(p',q')\geq \varphi(p',q)$. By single crossing, equality holds in both inequalities, and therefore $q \in \mathtt Q(p')$ and $q' \in \mathtt Q(p)$.  
Conversely, it is easily seen by taking $Q=\{q,q'\}$ that if $\mathtt Q$ defined as in~\eqref{Q-argmax} above is nonreversing for all $Q$, then $\varphi(p,q)$ has single crossing in (p,q).
}
\hfill\rule{.1in}{.1in}

\end{property}

Our next property connects nonreversingness with the notion of a \emph{P-function}.  Such functions arise in a number of applications, and are the subject of a rich literature (e.g., Varga \cite{Varga2000}). 

\begin{property}\emph{\textbf{A P-function is  nonreversing.}
An affine function $\mathtt q:\mathbb R^N\rightarrow \mathbb R^N$, written as $\mathtt q(p) = Ap+b$ for an $N\times N$ matrix $A$ and an $N\times 1$ vector $b$,  is a $P$ function if the matrix $A$ is a $P$ matrix (every principal minor is positive).}   

{\em More and Rheinboldt \cite[Definition 2.5, p. 49]{more1973p} generalize this notion.  
A function $\mathtt q:\mathbb R^N\rightarrow \mathbb R^N$ is a $P$-function if for every $p$ and $p'$ with $p\not= p'$, there exists $z\in\{1,\ldots,N\}$ such that  $(p_z-p_z')(\mathtt q(p)_z-\mathtt q(p')_z) >0$. 
An implication of this condition is
\[
p\ge p', \mathtt q(p)\le \mathtt q(p') \implies p=p'.
\]
The natural application of this concept to correspondences $\mathtt Q:P\rightrightarrows Q$ is that %
\[
[q\in \mathtt Q\left(p\right), q^{\prime}\in
\mathtt Q\left(p^{\prime}\right), q\leq q^{\prime}, p\geq p^{\prime}] \implies
p=p^{\prime}.
\]
The  $P$-correspondence property immediately implies nonreversingness.
\hfill\rule{.1in}{.1in}
}
\end{property}

Gale and Nikaido \cite{GandN1965} show that a smooth P-function defined on a rectangle is injective [Theorem 4.2, p. 86] and and establish conditions under which it has an isotone inverse [Theorem 5, p. 87]. P-functions and their properties are discussed at length in Nikaido \cite{Nikaido68}.

\subsection{M0-correspondences and Inverse Isotonicity}

\subsubsection{M0-correspondences and Related Notions}\label{crawl}

We introduce the notion of an \emph{M0-correspondences}:

\begin{definition}[M0-correspondence]
An \emph{M0-correspondence} is a correspondence which satisfies unified gross substitutes and is nonreversing.   
\end{definition}

\noindent  We view unified gross substitutes as the more substantive of the two conditions, with nonreversingness typically being innocuous. 

We noted in Property \ref{bigdog} that unified gross substitutes is preserved by aggregation.  The same is not the case once we add nonreversingness:

\begin{remark}[{\bf M0-correspondences fail to aggregate}]
Appendix \ref{par:sum-m-corresp} shows that the sum (as defined in Property~\ref{pty:agg-pres-ugs}) of two M0-correspondences $\mathtt Q$ and $\mathtt Q^\prime$ is not necessarily an M0-correspondence.
\end{remark}

If $\mathtt Q$ is an M0-correspondence, then neither $\mathtt Q$ nor $\mathtt Q^{-1} $ need be point-valued.  In the special case in which both are point-valued, we recover the existing notion of an \emph{M-function}: 

\begin{definition}[M-function]
An \emph{M-function} is an M0-correspondence $\mathtt Q$ which is point-valued and  has point-valued inverse $\mathtt Q^{-1}$. \label{shade}
\end{definition}

We can then mix and match to define:

\begin{definition}[M0-function]\label{rainbow}
An \emph{M0-function} is an M0-correspondence $\mathtt Q$ which is point-valued. \end{definition}

\begin{definition}[M-correspondence]
An \emph{M-correspondence} is an M0-correspondence $\mathtt Q$ which has point-valued inverse $\mathtt Q^{-1}$.
\end{definition}

\noindent We summarize with the following table:
\begin{center}
\begin{tabular}{|l||c|c|} 
 \hline
  ~ & $\mathtt Q^{-1} $ is point-valued & $\mathtt Q^{-1} $ is set-valued \\ 
  \hline
  \hline
$\mathtt Q$ is point-valued & $\mathtt Q$ is an M-function & $\mathtt Q$ is an M0-function \\ 
\hline
 $\mathtt Q$ is set-valued & $\mathtt Q$ is an M-correspondence & $\mathtt Q$ is an M0-correspondence \\ 
 \hline
\end{tabular}
\end{center}

This string of definitions is based on our characterization of M0-correspondences in terms of unified gross substitutes and nonreversingness.  
More and Rheinboldt \cite[Definition 2.3, p. 48]{more1973p} (see also Ortega and Rheinboldt \cite[Definition 13.5.7, p. 468]{OandR2000}) define an M-function $\mathtt Q$ as one that (in our terms) satisfies weak gross substitutes and the condition that $\mathtt Q(p)\le \mathtt Q(p')$ implies $p\le p'$.  
In Appendix \ref{water}, we show that our definition and More and Rheinboldt's definition of an M-function are indeed equivalent.

To motivate the labels for these various notions, we recall  that an  M-matrix is a square matrix with every off-diagonal entry less or equal than zero and with every principal minor greater than zero. An M0-matrix is a square matrix with every off-diagonal entry less or equal than zero and with every principal minor greater than {\em or equal to} zero.%
\footnote{Our definition of an $M$-matrix is common and matches (for example) that of Ortega and Rheinboldt \cite[Definition 2.4.7, p. 54]{OandR2000} (see Plemmons \cite[Theorem 1, p. 148]{Plemmons1977} for the equivalence).  Plemmons \cite{Plemmons1977} uses the terms nonsingular M-matrix and M-matrix in place of the M-matrix and M0-matrix terms invoked here.}
Then we have:

\begin{property}\label{demon}{\em {\bf M(M0)-matrices Induce M(M0)-functions.}
Let $Q$ be an $n \times n$ matrix and consider the function $\mathtt{Q}(p) = Qp$. Then $\mathtt{Q}$ is an M-function if and only if $Q$ is a M-matrix, and $\mathtt{Q}$ is a M0-function if and only if $Q$ is an M0-matrix.}

{\em To confirm this statement, we first note that $Qp$ is obviously point valued.  If $Q$ is an M-matrix, then it is nonsingular and hence has an inverse $Q^{-1}$, and in addition the nonsingular matrix $Q$ is the inverse of $Q^{-1}$, ensuring that $Q^{-1}$ is point valued.  Next, if $Q$ is an M0-matrix, then $Qp$ is again point-valued but $Q$ may be  singular, in which case $Q^{-1}$ can be set-valued.\hfill\rule{.1in}{.1in}}
\end{property}

\noindent  Our terms thus exte notions of an M-function and M0-function to nonlinear functions and correspondences.

We can give examples for each of these categories.  As Section \ref{sun} explains, Berry, Gandhi and Haile \cite{berry2013connected} offer sufficient conditions for a function to be an M-function, with one example (among others) being the estimation of the choice probabilities in a random-utility discrete choice model.  Section \ref{apple} below presents examples of M0-correspondences. The excess supply correspondence of a competitive exchange economy with strictly convex preferences is in general an M0-function.  Each price vector gives rise to a unique excess supply, but the inverse may be set-valued---multiple price vectors may give rise to (for example) an excess supply of zero.  Relaxing strict convexity to convexity but adding weak gross substitutes yields an M-correspondence.  A price vector may give rise to multiple excess supplies, but the inverse is point-valued---each excess supply comes from a unique price vector. The following properties present additional examples.

\begin{property}\label{hydra}\emph{\textbf{The aggregate supply function associated with a logit model without price normalization is an M0-function.}} {\em
		Consider a logit model where there are $n_x$ producers of type $x\in\mathcal{X}$ and each producer $x$ produces one unit of good $z$ chosen from a set of goods $\{1,...,N\}$.  Producer type $x$ gleans the random profit  $\pi_{xz}(p_z)+\varepsilon_z $ from producing good $z$, where $\pi_{xz}(.)$ is increasing and $(\varepsilon_z)_{z\in \{1,...,N\}}$ is an independent-and-identically-distributed random vector with Gumbel distribution. Then the aggregate supply function is a function $\mathtt Q_z: \mathbb{R}^N \to \mathbb{R}^N$
		given by   
		$$
		\mathtt Q_z (p) = \sum_{x\in \mathcal{X}} n_x \frac {\exp(\pi_{xz}(p_z))}  { \sum_{ z^{\prime}=1}^ N \exp(\pi_{xz^{\prime}}(p{z^\prime}))},  
		$$ 
		and is an M0-function.\hfill\rule{.1in}{.1in}}
\end{property}

\begin{property}\emph{\textbf{The aggregate supply function associated with a logit model with price normalization is an M-function.}}
	{\em 
		In the model of Property \ref{hydra}, assume that the price $p_N$ of good $N$ is normalized to equal $\pi$.  The aggregate production function restricted and corestricted to $\mathbb{R}^{N-1 }$, denoted by $\mathtt Q_z: \mathbb{R}^{N-1 }\to \mathbb{R}^{N-1}$
		and given by   
		$$\mathtt Q_z (p) = \sum_{x\in \mathcal{X}} n_x \frac {\exp(\pi_{xz}(p_z))}  { \exp( \pi_{xN}( \pi ) )+\sum_{ z^{\prime}=1}^ N \exp(\pi_{xz^{\prime}}(p{z^\prime}))},  $$ 
		is an M-function. \hfill\rule{.1in}{.1in}}
\end{property}

\subsubsection{Inverse Isotonicity Theorems}\label{grenade}

We introduce the following condition, which is slightly stronger than requiring that the inverse correspondence $\mathtt Q^{-1}$ is isotone in the strong set order (a.k.a. Veinott's order, Veinott \cite{veinott1992lattice}).%
\footnote{Isotone correspondences are sometimes also said to be (weakly) increasing, and contrast with antitone (or weakly decreasing) correspondences. }

\begin{definition}[{\bf Totally Isotone Inverse}]\label{def:Veinott-isotone}
A correspondence $\mathtt Q:P\rightrightarrows Q$ has {\em totally isotone inverse} if, whenever $q \in \mathtt Q \left( p\right) $ and $q ^{\prime }\in \mathtt Q\left( p^{\prime}\right) $ 
are such that there exists $B\subseteq \{1,\ldots,N\}$ with $p_{z} \leq p_{z}^{\prime }$ for all $z\in B$ and $q_{z}\leq q_{z}^{\prime }$ for all $z\notin B$, we have
\begin{equation}\label{parking}
q\in \mathtt Q\left( p\wedge p^{\prime }\right) \text{ and }q^{\prime }\in \mathtt Q\left(
p\vee p^{\prime }\right).
\end{equation} 
\end{definition}

If we require \eqref{parking} to hold only for the case  $B=\emptyset$, then the inverse correspondence $\mathtt Q^{-1}$ is isotone in the strong set order, i.e,  whenever $q \in \mathtt Q \left( p\right) $ and $q ^{\prime }\in \mathtt Q\left( p^{\prime}\right) $ 
are such that  $q_{z}\leq q_{z}^{\prime }$ for all $1\leq z \leq N$, we have
\begin{equation*}
	q\in \mathtt Q\left( p\wedge p^{\prime }\right) \text{ and }q^{\prime }\in \mathtt Q\left(
	p\vee p^{\prime }\right).
\end{equation*} 
In this case we say simply that $\mathtt Q$ is {\em inverse isotone}, a weaker property obviously implied by totally isotone inverse.  
We can equivalently express the inverse isotonicity of $\mathtt Q$, or equivalently the isotonicity of $\mathtt Q^{-1}$ in the strong set order,  as the requirement that  whenever $p \in \mathtt Q^{-1} \left( q\right) $ and $p^{\prime }\in \mathtt Q^{-1} \left( q^{\prime}\right) $ where $q\le q'$,  it follows that 
\begin{equation*}
	p\wedge p^{\prime } \in \mathtt \mathtt Q^{-1}\left(q\right) \text{ and } p\vee p^{\prime }\in \mathtt \mathtt Q^{-1}\left( q^{\prime }\right) \text{.}
\end{equation*}

The following inverse isotonicity result is the building block for subsequent applications.  

\begin{theorem}\label{belgian}
	Let $\mathtt Q:P\rightrightarrows Q$ satisfy unified gross substitutes. Then the following conditions are equivalent:
	
	(i) $\mathtt Q$ is nonreversing  (i.e.,  $\mathtt Q$ is a M0-correspondence), and 	
	
	(ii) $\mathtt Q$ has totally isotone inverse.
	
\end{theorem}

\paragraph{Proof} Assume (i) and consider $q \in \mathtt Q \left( p\right) $,  $q ^{\prime }\in \mathtt Q\left( p^{\prime}\right) $, and $B\subseteq \{1,...,N\}$ such that $p_{z} \leq p_{z}^{\prime }$ for all $z\in B$ and $q_{z}\leq q_{z}^{\prime }$ for all $z\notin  B$. By unified gross substitutes, one has the existence of $q^\wedge \in \mathtt Q \left( p \wedge p^\prime \right) $ and $q^\vee \in \mathtt Q \left( p \vee p^\prime \right) $ such that
\begin{eqnarray*}
p_{z} \leq p_{z}^{\prime }&\implies& q_{z}\leq q_{z}^{\wedge } \\
p_{z} >p_{z}^{\prime }&\implies &q_{z}^{\prime }\leq q_{z}^{\wedge }.
\end{eqnarray*}%
We have $p_{z}>p_{z}^{\prime }\implies z \notin B$, which implies $q_{z}\leq q_{z}^{\prime
} $.  Combining with the inequalities above, we see that 
$q_z\leq q_z^{\wedge }$ holds for any $1 \leq z \leq N$. But then $q\leq q^{\wedge }$ and $p\geq p\wedge
p^{\prime }$, so by nonreversingness it follows that $q\in \mathtt Q\left( p\wedge
p^{\prime }\right) $. A similar reasoning shows that $q^{\prime }\in \mathtt Q\left( p\vee p^{\prime
}\right) $.  We have therefore shown statement (ii). 

Conversely, we assume statement (ii) holds and show  $\mathtt Q$ is nonreversing. Take $p\geq p^\prime$ and $q \leq q^\prime$ such that $q \in \mathtt Q(p)$ and $q^\prime \in \mathtt Q(p^\prime)$. Letting $B = \emptyset$, because $\mathtt Q$ has totally isotone inverse, we have that $q \in \mathtt Q(p \wedge p^\prime) $, which is equivalent to $q \in \mathtt Q (p^\prime)$, and $q^\prime \in \mathtt Q(p \vee p^\prime) $, which is equivalent to $q^\prime \in \mathtt Q (p)$, which shows (i).
\hfill\rule{.1in}{.1in}

\bigskip

Appendix \ref{rooster} generalizes this result to partial inverse correspondences.  It is an immediate implication that: 

\begin{corollary}\label{gobble}
Let $\mathtt Q$ be an M0-correspondence.  Then the set of prices $\mathtt Q^{-1}(q)$ associated with an allocation $q$ is a sublattice of $P$.
\end{corollary}

\paragraph{Proof} Take $p\in \mathtt Q^{-1}\left(q\right)$ and $%
p^{\prime}\in \mathtt Q^{-1}\left(q\right)$. Then $q\leq q$ yields $%
q\in\mathtt Q\left(p\wedge p^{\prime}\right)$ and $q^{\prime}\in\mathtt Q\left(p%
\vee p^{\prime}\right)$.\hfill\rule{.1in}{.1in}

\bigskip

The inverse image $\mathtt Q^{-1}(\tilde q)$ will often describe an equilibrium.  For example, $\mathtt Q$ may be the aggregate supply function of a competitive economy and $\tilde q$ may be the negative if the aggregate endowment, so that $\mathtt Q^{-1}(\tilde q)$ identifies the set of competitive equilibrium prices.  Theorem \ref{belgian} thus gives us monotone comparative statics results for equilibrium problems, in our example describing how competitive equilibrium prices vary in the endowment. 

One can can show that under uniform gross substitutes, $\mathtt Q$ is a M-correspondence if and only if it is totally nonreversing. 

\begin{theorem}\label{thm:strong-inverse-isotonicity}

Let $\mathtt Q:P\rightrightarrows Q$ satisfy unified gross substitutes. Then the
following three conditions are equivalent:

(i) $\mathtt Q$ is an M-correspondence;

(ii) $\mathtt Q^{-1}$ is point-valued and isotone where not empty, i.e. $q \in \mathtt Q\left( p\right)$,  $ q^\prime \in \mathtt Q\left( p^\prime\right)$

\quad  and $ q \leq  q^\prime$ imply $p\leq
p^{\prime }$;

(iii) $\mathtt Q$ is strongly nonreversing.

\end{theorem}

\paragraph{Proof}  \underline{(i) implies (ii)}: Assume $\mathtt Q$ is a M-correspondence and assume
 $q \in \mathtt Q\left( p\right)$, $ q^\prime \in \mathtt Q\left( p^\prime\right)$ and $ q \leq  q^\prime$. By Theorem~\ref{belgian}, we have $ q \in \mathtt Q\left( p\wedge p^{\prime }\right) $ and $ q^\prime \in \mathtt Q\left( p\vee p^{\prime }\right) $. Because $\mathtt Q$ is
injective, it follows that $p=p\wedge p^{\prime }$ and thus $p\leq p^{\prime
}$.

\underline{(ii) implies (iii)}: Assume $\mathtt Q$ is inverse isotone and assume $\mathtt q\left(
p\right) \leq \mathtt q\left( p^{\prime }\right) $ and $p\geq p^{\prime }$. By
inverse isotonicity one has $p\leq p^{\prime }$, and thus $p=p^{\prime }$.

\underline{(iii) implies (i)}: Assume $\mathtt Q$ is strong nonreversing. Then it is nonreversing  and thus $\mathtt Q $
is a M0-correspondence. Assume $q\in \mathtt Q(p)$ and $q\in \mathtt Q(p')$.  
By unified gross substitutes, we have $q^{\vee}\in \mathtt Q(p\vee p')$ with $q^{\vee}\le q$.  Because $p\vee p'\ge p,p'$, strong nonreversingness  gives $p = p\vee p' = p'$.  Hence, 
%
%
$p=p^{\prime }$ and
thus $\mathtt q$ is injective.
\hfill\rule{.1in}{.1in}

\section{Connections with Existing Theories}
\label{sect:connections}

\subsection{Kelso and Crawford's Gross Substitutes}\label{par:kelso-crawford}

We can connect unified gross substitutes to Kelso and Crawford's well-known definition \cite[p. 1486]{KandC1982} of substitutes.  Translating their definition to our setting of supply correspondences, the correspondence $\mathtt Q$ has the {\em Kelso-Crawford substitutes} property if, given two price vectors $p$ and $p'$ with  $p'\le p$, for any $q\in \mathtt Q\left(p\right)$ there exists $q'\in \mathtt Q(p')$ such that $p_z=p_z' \implies q'_z\ge q_z$.  

\begin{property}\emph{\textbf{Unified gross substitutes implies Kelso-Crawford substitutes.}} It follows immediately from  \eqref{lester} and the definition of Kelso-Crawford substitutes that 
unified gross substitutes implies the Kelso and Crawford substitutes property.
\end{property}

\noindent Appendix \ref{par:kc-example} shows that Kelso and Crawford's definition does not imply unified gross substitutes.

We refer to our notion as {\em unified gross substitutes} in order to emphasize the distinction between Kelso and Crawford gross substitutes and our notion, namely that the latter ``unifies'' the conditions \eqref{lester} and \eqref{flatt} by requiring them to be satisfied by 
common $q^{\wedge}$ and $q^{\vee}$.  Kelso and Crawford's condition implies there are allocations $q^{\wedge}$ and $q^{\vee}$ satisfying  \eqref{lester} and also allocations $q^{\wedge}$ and $q^{\vee}$ satisfying  \eqref{flatt}, but allows these allocations to differ.

\subsection{Polterovich and Spivak's Gross Substitutes}\label{wobble}

 Polterovich and Spivak \cite[Definition 1, p. 118]{Posp1983} propose a notion of gross substitutes for correspondences.   (See Howitt \cite{Howitt80} for an intermediate notion.)  In our notation and setting, the correspondence $\mathtt Q(p)$ satisfies their notion of gross substitutability if, for any price vectors $p\le p'$ and any $q\in \mathtt Q(p)$ and $q'\in \mathtt Q(p')$, it is {\em not} the case that
\[
q'_z>q_z~~\forall z~s.t.~p_z=p'_z.
\]
Polterovich and Spivak's notion thus stipulates that if the prices of some set of goods increase while others remain constant, it cannot be the case that every one of the quantities associated with the latter set strictly increases.  

Polterovich and Spivak \cite[Lemma 1, p. 123]{Posp1983} show that if the correspondence $\mathtt Q$ maps from the interior of $\mathbb R^N_+$ into $\mathbb R^N$,  is convex valued and closed valued, and maps compact sets into nonempty bounded sets, then their gross substitutability condition implies \eqref{lester}--\eqref{flatt}.   Appendix \ref{par:polterovich-comparison} shows that our requirement that $P$ be a sublattice of $\mathbb{R}^N$ does not suffice for this result, and that in general neither notion implies the other. Polterovich and Spivak \cite[p. 119]{Posp1983} note that their definition of  gross substitutes for correspondences is not preserved under the addition of correspondences, unlike unified gross substitutes.  

\subsection{Berry, Gandhi and Haile's Connected Strict Substitutes}\label{sun}

Berry, Gandhi and Haile \cite{berry2013connected}, abbreviated as BGH, examine functions $\mathtt q:P\rightarrow Q$. They make the follwoing assumptions, which allow them to show that an inverse function is isotone and point-valued:

\bigskip

 \textbf{BGH, Implicit Assumption 0}: $\mathtt q$ is point-valued.

\smallskip
 \textbf{BGH, Assumption 1}: $\mathtt q$ is defined on a Cartesian product of sets.

\smallskip

 \textbf{BGH, Assumption 2.a}: $\mathtt q$ has the gross substitutes property. 

\smallskip

\textbf{BGH, Assumption 2.b}: The function defined by
\[
\mathtt q_{0}\left( p\right) :=-\sum_{z=1}^{N}\mathtt q_z\left(
p\right) 
\] 

\quad is weakly decreasing in each $p_{z}$ for $z\in \left\{
1,...,N\right\} $.

\smallskip

\textbf{BGH, Assumption 3}: 
%
%
%
For all $B\subseteq \{1,\ldots,N\}$, $p$ and $p'$ such that  $p_z=p_z'$

\quad for all 
$z \in B$ and $p_z>p_z'$ for all $z\not \in B$, there exists $\tilde z\in B\cup\{0\}$ such that 

\quad $\mathtt q_{\tilde z}(p)<\mathtt q_{\tilde z}(p')$.

%

\bigskip 

We find it useful to record Berry, Gandhi and Haile's restriction to demand functions as an implicit Assumption 0.  Berry, Gandhi and Haile's Assumption 1, that 
$\mathtt q$ is defined on a Cartesian produce of Euclidean spaces  \cite[p. 2094]{berry2013connected}, implies our Assumption \ref{swat} that $P$ is a sublattice of $\mathbb{R}^N$. Their second assumption \cite[p. 2094]{berry2013connected}  (in our notation, and translating from their framing in terms of demand functions to our framing in terms of supply functions) is split here into Assumptions by 2.a and 2.b. Lastly, their  Assumption 3 \cite[p. 2095]{berry2013connected} is a connected strict substitutes assumption which expresses that one cannot partition the set of goods into two set of products such that no good in the first set is a substitute for some good in the second set.  We have stated this assumption in the equivalent form established in their Lemma 1.    

Applying
our Theorem~\ref{belgian} gives a variant of Berry, Gandhi and Haile's \cite{berry2013connected} Theorem 1 which operates under weaker assumptions (more precisely, without their Assumption 3) but delivers a weaker conclusion:%
\footnote{Berry, Gandhi and Haile establish the stronger inverse isotonicity notion that appears later in Corollary~\ref{church}.}
\begin{corollary}\label{back}
Under Berry, Gandhi and Haile's Assumptions 0, 1, 2.a, and 2.b, the function $\mathtt q$ is an M0-function and hence the inverse  $\mathtt q^{-1}$ of $\mathtt q$ is isotone in the strong set order, that is $\mathtt q(p) \leq \mathtt q(p^\prime) $ implies $\mathtt q(p)= \mathtt q(p\wedge p^ \prime)$ and $\mathtt q(p^\prime)= \mathtt q(p\vee p^ \prime)$. 
\end{corollary}

\paragraph{Proof}
Recall that 
Property \ref{prop:wgs-implies-ugs} establishes that if the function $ {\mathtt q}$ satisfies weak gross substitutes, which it does under Assumption 2.a, then it satisfies unified gross substitutes. Using Assumption 2.b and Property~\ref{pty:mto-implies-nr}, we get that $\mathtt q$ is also nonreversing, and hence $\mathtt q$ is an M0-function.
\hfill\rule{.1in}{.1in}

\bigskip

Applying our Theorem~\ref{thm:strong-inverse-isotonicity} now gives   an alternative route to Berry, Gandhi and Haile's Theorem \cite{berry2013connected} Theorem 1 and Corollary 1.  

\begin{corollary}\label{church}
Under Berry, Gandhi and Haile's Assumptions 0, 1, 2.a, 2.b and 3, the function $\mathtt q$ is an $M$-function and hence is inverse isotone ($\mathtt q(p) \leq \mathtt q(p^\prime) $ implies $p \leq p^\prime$) and $\mathtt q^{-1}$ is point-valued. 
\end{corollary}

\paragraph{Proof} We show that under the additional Assumption 3, the function $\mathtt q$ is in addition strongly nonreversing. Indeed, assume $\mathtt q(p) \leq  \mathtt q (p^\prime)$ and $p \geq p^\prime$. By nonreversingness, one has $\mathtt q(p) = \mathtt q(p^\prime)$. Assume $p > p^\prime$. 
Let $B=\{z:p_z=p'\}\ne \{1,\ldots, N\}$.  By BGH Assumption 3, there exists $z\in B\cup \{0\}$ such that $\mathtt q_z(p)<\mathtt q_z(p')$, a contradiction.  
%
%
%
Hence $p = p^\prime$, giving strong nonreversingness.  Theorem~\ref{thm:strong-inverse-isotonicity}  then gives the result.
\hfill\rule{.1in}{.1in}

\bigskip

We can generalize Berry, Gandhi and Haile's result to correspondences.  

\bigskip

\textbf{BGH, Assumption 3'}: 
For all $B\subseteq \{1,\ldots,N\}$, $p$ and $p'$ such that 
$p_z=p_z'$ 

\quad for all $z \in B$ and 
$p_z>p_z'$ for all $z\not \in B$, and all $q\in \mathtt Q(p)$ and $q'\in \mathtt Q(p')$

\quad with $q\le q'$, there exists $\tilde z\in B\cup\{0\}$ such that $q_{\tilde z}<q'_{\tilde z}$.

\bigskip

\noindent Let $\tilde {\mathtt Q}:\{1\}\times P\rightrightarrows \mathbb R\times Q$ be the correspondence, written $\tilde q\in \tilde {\mathtt Q}(\tilde p)$ constructed from $Q$ by letting $\tilde p=(1,p)$ and $\tilde q = (-\sum_{z=1}^Nq_z,q)$.   
Then $\tilde {\mathtt Q}$ by construction satisfies constant aggregate output and hence nonreversingingss.  An argument analogous to the proof of Corollary \ref{church} establishes that $\mathtt Q$ is strongly nonreversing.  
Theorem \ref{thm:strong-inverse-isotonicity} then immediately gives:

\begin{corollary}
Let $\tilde {\mathtt Q}$ satisfy unified gross substitutes. Then $\tilde {\mathtt Q}$ is an $M$-function and hence is inverse isotone ($\mathtt q(p) \leq \mathtt q(p^\prime) $ implies $p \leq p^\prime$) and and $\mathtt Q^{-1}$ is point-valued.      
\end{corollary}

\noindent Appendix \ref{horns} shows that if the correspondence $\tilde {\mathtt Q}$ satisfies unified gross substitutes, then so does $Q$.

\subsection{Topkis' Theorem}
\label{par:connections-mcs}
When the inverse $\mathtt Q^{-1}$ is itself the solution to a maximization problem, we may be able to obtain  an inverse isotonicity result via monotone comparative static arguments familiar from Topkis \cite{Topkis1978,Topkis1998}.    For example, we may be interested in the supply function of a competitive, multiproduct firm, given by 
\begin{equation}\label{saxophone}
\mathtt Q(p) = \arg \max_{q\in \mathbb{R}_{+}^{N}}\left\{p^{\top }q-c\left( q\right) \right\}=\partial c^*(p),
\end{equation}		
where $c$ is a convex cost function $c:\mathbb R_+^N\rightarrow \mathbb R$,  $c^*(p) = \max_{q\in \mathbb{R}_{+}^{N}}\left\{p^{\top }q-c\left( q\right) \right\}$ is the indirect profit function, and $\partial c^*(p)$ is the  subdifferential of $c^*$ at $p$.%
\footnote{Given the convexity of the cost function $c$, the general form of Shephard's lemma (Rockafellar \cite[Theorem 23.5, p. 218]{Rocket1997}) gives $\mathtt Q(p) = \partial c^*(p)$, where $\partial c^*(p)$ is the subdifferential of the indirect profit function.}
%
%
%
%
\[
\mathtt Q^{-1}(q)=\arg \max_{p\in \mathbb{R}_{+}^{N}}\left\{p^{\top }q-c^*\left( p\right) \right\} ~=~ \partial c(q),
\]
where $\partial c(q)$ is the subdifferential of $c$ at $q$.

We now have two routes to establishing the isotonicity in the strong set order of $\mathtt Q^{-1}$.   Section \ref{king} pursues an approach based on establishing an equivalence between the submodularity of the profit cost function $c^*$ and the supply correspondence $\mathtt Q = \partial c^* $ satisfying unified gross substitutes.  Alternatively,  we can note  that the function $p^{\top }q-c^*\left( p\right)$ is supermodular in $p$ (given the submodularity of $c^*$) and exhibits increasing differences, and so $\mathtt Q^{-1}$ is increasing in the strong set order (Topkis \cite[Theorem 6.1, p. 317]{Topkis1978}).

More generally, given a correspondence $\mathtt Q:P\rightrightarrows Q$, we can potentially exploit the existence of a function $g(p,q)$ such that the inverse correspondence  $\mathtt Q^{-1}$ is the solution to the maximization problem
\begin{equation}\label{string}
	\mathtt Q^{-1}(q) = \arg\max_{p\in P}g(p,q).
\end{equation}
%
%
%
%
%
%
Topkis \cite{Topkis1978} assumes that $g$ is supermodular in $p$ and has increasing differences in $(p,q)$, which imply the following sufficient condition for his isotonicity result:
\begin{equation}\label{spinner}
	q\geq q^{\prime}
	\implies
	g(p\vee p',q)-g(p',q)
	\geq
	g(p,q')-g(p\wedge p^{\prime},q').
\end{equation}
Milgrom and Shannon's \cite{MandS94} weaker assumptions of quasi-supermodularity and single crossing similarly imply
\begin{equation}\label{rooster}
	q\geq q^{\prime}\implies\left\{
	\begin{array}{c}
		g\left(p,q'\right)-g\left(p\wedge p^{\prime},q'\right)\geq0\implies
		g\left(p\vee
		p^{\prime},q\right)-g\left(p^{\prime},q\right)\geq0 \\
		g\left(p\vee
		p^{\prime},q\right)-g\left(p^{\prime},q\right)\leq0%
		\implies g\left(p,q'\right)-g\left(p\wedge p^{\prime},q'\right)\leq0.
	\end{array}%
	\right.
\end{equation}
%
%
%
%
Both sets of assumptions lead to monotone comparative statics results. In Appendix \ref{buffalo} we give two examples showing that our (unified gross substitutes and nonreversingness) conditions for inverse isotonicity do not imply those of Topkis or Milgrom and Shannon, nor do the reverse implications hold.  We thus have independent conditions for the case of optimization problems, while  we can also view our results as extending the analysis beyond the purview of optimization problems.

\section{Applications}\label{apple}

\subsection{Profit Maximization}\label{king}
In keeping with our interpretation of $\mathtt Q$ as a supply correspondence, suppose a competitive multiproduct firm faces output price vector $p\in \mathbb R^N$ and convex cost function $c:\mathbb R^N\rightarrow \mathbb R$.  Given a price $p$, the set of optimal production vectors $\mathtt Q(p)\subseteq \mathbb R^N$ is given by
\eqref{saxophone}.
%
%
%
%
%
%
The conjugate of the cost function $c$, denoted by $c^*(p)$, is the indirect profit function:
\[
c^*(p) = \max_{q\in \mathbb R^N}\left\{p^Tq-c(q)\right\}.  
\]
Note that by Property \ref{indiff}, the $\arg\max$ correspondence $\mathtt{Q}$ is nonreversing, and hence $\mathtt{Q}$ is a M0-correspondence if and only if it satisfies unified gross substitutes.

The following result, whose proof is given in Appendix \ref{purple}, relates unified gross substitutes to the submodularity of the indirect profit function. It offers a continuous counterpart to Ausubel and Milgrom's \cite{AandM2002} Theorem 10:

\begin{theorem}\label{spice}\label{wonder}
\label{thm:equivalence}The following conditions are equivalent:\\
(i) the indirect profit function $c^{\ast }$ is submodular, and\\
(ii) the supply correspondence $\mathtt Q(p) = \partial c^{\ast }\left( p\right) $ satisfies unified gross substitutes, and\\
(iii) the supply correspondence $\mathtt Q(p) =  \partial c^{\ast }\left( p\right) $ is an M0-correspondence.
\end{theorem}

\begin{remark}\label{remark_imperfect_competition}
{\em
We can relax the assumption of perfect competition and accommodate a more general revenue function.   Consider a firm whose profit maximization problem is given by
\begin{equation}\label{move}
c^\pi(p) = \max_{q\in\mathbb R^N} \left\{ \sum_{z=1}^N\pi_z(p_z,q_z)-c\left(q\right) \right\},
\end{equation}
where $c:\mathbb R^N \rightarrow \mathbb R$ is an increasing and strictly convex cost function and each function $\pi_z:\mathbb R^2\rightarrow\mathbb R$ is a revenue function that is differentiable, increasing in $p_z$ and concave in $q_z$ and satisfies $\frac{\partial^2\pi_z(p_z,q_z)}{\partial p_z \partial q_z}>0$ and $\frac{\partial^3\pi_z(p_z,q_z)}{\partial p_z \partial q_z^2 }>0$.  The interpretation is that a firm chooses quantities $q_1\ldots, q_N$ to sell in each of $N$ markets.  The total cost of production is given by $c(q)$.  Revenue in each market $z$ is given by $\pi_z(p_z,q_z)$.   Market $z$ may be perfectly competitive, in which case we can take $p_z$ to be the price in the market.  More interestingly, market $z$ may be imperfectly competitive, in which case $p_z$ is a demand shifter, with higher values of $p_z$ indicating higher revenue and higher marginal revenue for each given quantity. 

Appendix \ref{sawbones} shows that the supply correspondence $\mathtt Q$ emanating from \eqref{move} satisfies unified gross substitutes if and only if the indirect profit function $c^{\pi}$ is submodular, providing an extension of Theorem~\ref{spice} to revenue functions $\pi_z(p_z,q_z)$ that are more general than $\pi_z(p_z,q_z) = p_z q_z$.
\hfill\rule{.1in}{.1in}}
\end{remark}

\begin{remark}{\em
Gul and Stacchetti  \cite[Theoren 1, p, 99]{GandS1999} show that in their context of utility maximization with indivisible goods (so that an individual consumes either zero or one unit of each good and a consumption bundle is a list of which goods are consumed), gross substitutes is equivalent to each of two additional properties that they introduce, namely no complementarities and single improvement.  Their no-complementarities condition stipulates that if $X$ and $Y$ are both optimal consumption bundles and $Z\subset X\setminus Y$, then one can also obtain an optimal consumption bundle by removing the goods $Z$ from the set $X$ and adding some elements of $Y$.  For the case of divisible goods, Galichon et al. \cite{Galichon2022} offer the following continuous counterpart, which requires the optimality of only one of the two bundles:%
\footnote{A counterpart of Gul and Stacchetti's single improvement property is less relevant in the case of divisible goods.  Their single improvement property states that if the set of goods $X$ is suboptimal, then one can find a superior set by at most deleting one existing object from $X$ and adding one additional object.  The counterpart of this in the divisible case would presumably be that if the allocation $q$ is suboptimal, then one can reach a superior allocation by adjusting at most two dimensions of $q$.  Suppose, however that we maintain the standard assumption that the cost function $c$ is convex.  If  $c$ is differentiable, then it is immediate that for any suboptimal $q$, there is an adjustment of a single dimension of $q$ that leads to a superior allocation.  We thus get the counterpart of single improvement without any appeal to substitution properties.}	

\begin{definition}[{\bf No Complementarities}]\label{def:no-complement}
	The correspondence $\mathtt Q(p)$ satisfies the no complementarities property if whenever $q\in \mathtt Q(p)$ and $q'\neq q$, then for each $\delta _{1}$ such that $0\leq
	\delta _{1}\leq \left( q^{\prime }-q\right) ^{+}$, there exists $\delta _{2}$
	such that $0\leq \delta _{2}\leq \left( q-q^{\prime }\right) ^{+}$ such that
	\[
	p^{\top} (q'+\delta_2-\delta_1)-c(q'+\delta_2-\delta_1)\ge p^{\top} q'-c(q').
	\]
\end{definition}

\noindent  When applied to the special case in which $q, q'\in \mathtt Q(p)$, the result is the functional equivalent of Gul and Stacchetti's condition, in that we begin with the optimal allocations $q$ and $q'$, and then construct a new optimal allocation by adding to $q'$ in some dimensions $z$ in which $q'_z<q_z$ and subtracting from $q'$ in some dimensions $z$ in which $q'_z>q_z$. Galichon et al. \cite{Galichon2022} prove the following:%
\footnote{This result is related to Chen and Li \cite[Theorem 2, p. 13]{ChenLi2020}  who show that if $c$ is convex and satisfies their S-EXC property, then  its conjugate is submodular.  However, Chen and Li's definition does not one allow to establish the converse implication, while the notion introduced in Definition~\ref{def:no-complement} allows us to formulate a necessary and sufficient condition.}

\begin{theorem}\label{tarantula}
	Assume $c\left( q\right) $ is convex.  The  correspondence $\mathtt Q$ defined in \eqref{saxophone} satisfies unified gross substitutes if and only if it satisfies no complementarities.  
\end{theorem}
\hfill\rule{.1in}{.1in}}
\end{remark}


\subsection{Structures of Solutions}

It is a familiar result that if each agent's (point-valued) demand function in an exchange economy satisfies strict gross substitutes, then that economy has a unique equilibrium (e.g., Arrow and Hahn \cite[Chapter 9]{AandH1971}).  In more general settings that allow for set-valued demand and supply functions, the proximate notion of a unique equilibrium is that the set of equilibrium prices forms a lattice.   Since Walras law implies nonreversingness (Property \ref{pty:walras-law}), we have the result that if the excess supply correspondence $\mathtt Q$ satisfies unified gross substitutes, then the set of equilibrium prices is a sublattice.  Polterovich and Spivak \cite[Corollary 1, p. 125]{Posp1983} similarly show that the set of equilibrium prices in an exchange economy satisfying their gross substitutes condition is a lattice.  
Gul and Stacchetti \cite[Corollary 1, p. 105]{GandS1999} have a similar result for economies with indivisible goods.

 In the point-valued case, the notion of \emph{subsolutions} (a price vector $p$ such that $Q(p) \leq 0$ (or any other exogenously-specified value) and \emph{supersolutions} (a price vector $p$ such that $Q(p)\geq 0$) play an important role. One can show in particular that if $Q$ is continuous and satisfies weak gross substitutes, then any subsolution $p$ that is not a solution, is such that there exists another subsolution $p^\prime$ such that $p < p^\prime$. As a result, if an iterative method produces a sequence of subsolutions $p_k$ that increases ``enough,'' in the sense that the set of subsolutions that dominate $p_k$ shrinks, one may show that it converges to a solution. As our larger goal is to provide iterative methods for correspondences, we investigate set-valued extensions of these notions.   Appendix \ref{wolverine} reports first steps in this broader research agenda.

\subsection{An Equilibrium Flow Problem}

\paragraph{Network.} Consider a network $\left(\mathcal{Z},\mathcal{A}\right)$ where $%
\mathcal{Z}$ is a finite set of nodes and $\mathcal{A} \subseteq \mathcal Z \times \mathcal Z$ is the set of directed
arcs.  If $xy \in \mathcal A$, we say that $xy$ is the arc from $x\in \mathcal Z$ to $y\in \mathcal Z$, and we say that $x$ is the starting point of the arc, while $y$ is its end point. We assume that there is no arc in $\mathcal A$ whose starting point coincides with the end point. We describe the network with an $|\mathcal{A}|\times
|\mathcal{Z}|$ \emph{arc-node incidence matrix} matrix $\nabla$, defined by letting,
for $xy\in \mathcal{A}$ and $z\in Z$
\[
\nabla _{xy,z}=\mathbf 1_{\left\{ z=y\right\}} -\mathbf 1_{\left\{ z=x\right\}}.
\]
We thus have $\nabla _{xy,z}=1$ if $xy$ is an arc  ending at $z$, and $\nabla _{xy,z}=-1$
if $xy$ is an arc beginning at $z$.  Otherwise $\nabla _{xy,z}=0$.

\paragraph{Prices, connection functions.} Let $p\in\mathbb{R}^{\mathcal{Z}}$ be a price vector, where we interpret $p_{z}$ as the price at node $z$.  
To have a concrete description, though we do not require this interpretation, one may consider a trader operating on node $xy$, who is able to purchase one unit of a commodity at node $x$, ship it along arc $xy$ toward node $y$, and resell it at node $y$. Given the resale price at node $y$, there is a certain threshold value of the price at node $x$ such that the trader is indifferent between engaging in the trade or not. This value is an increasing and continuous function of the price at node $y$, and can be expressed as $G_{xy}(p_y)$, where for each arc $xy \in \mathcal A$, the function 
$G_{xy}: \mathbb R \to \mathbb R$ is continuous and increasing, and called the \emph{connection function}.\footnote{This is related to the idea of a Galois connections, which explains the choice of the letter $G$;  see~N\"{o}ldeke and Samuelson \cite{NandS2017}.}
Hence, if $p_x>G_{xy}(p_y)$, the purchase price at node $x$ is excessive, and the trader will not engage in the trade. On the contrary, if $p_x<G_{xy}(p_y)$, the purchase price is strictly below indifference level and positive rent can be made from the trade on the arc $xy$.

Our framework allows for any situation where the per-unit rent $\Pi_{xy}(p_x,p_y)$ of the trade on arc $xy$ is a continuous and possibly nonlinear function of $p_x$ and $p_y$, increasing in the resale price $p_y$ and decreasing in the purchase price $p_x$. In that case, $G_{xy}(p_y)$ is implicitly defined from $\Pi_{xy}$ by $G_{xy}(p_y) = \Pi_{xy}(.,p_y)^{-1}(0)$, or equivalently, $\Pi_{xy}(G_{xy}(p_y),p_y)=0$.

\begin{example}[Additive case]\label{ex:additive}{\em 
A simple example of connection function  assumes linear surplus for the trader, in which case the per-unit profit of the trader on arc $xy$ is $p_y - p_x - c_{xy}$ where $c_{xy}$ is the unit shipping cost from $x$ to $y$, and the indifference price at arc $x$ is \begin{equation}\label{TU-connection}
G_{xy}(p_y) = p_y - c_{xy}. 
\end{equation}
We refer to this as the \emph{additive case}.   
\hfill\rule{.1in}{.1in}}
\end{example}

\paragraph{Exiting flow, internal flow, mass balance.} Let $q\in \mathbb R^{\mathcal Z}$ with $\sum _{z \in \mathcal Z}q_z=0$ attach a net flow to each node $z\in \mathcal Z$.  If $q_z>0$, then the net quantity $|q_z|$ must flow into node $z$, while $q_z<0$ indicates that the net quantity $|q_z|$ must flow away from node $z$. Hence, we call $q$  the vector of \emph{exiting flows}. We let  $\mu\in\mathbb R^{\mathcal A}_+$ be the vector of \emph{internal flows} along arcs, so that  $\mu_{xy}$ is the flow through arc $xy$.  The feasibility condition connecting these notions is that, for any $z\in\mathcal{Z}_{0}$, the total internal flow that arrives at $z$ minus the total internal flow that leaves $z$ equals the exiting flow at $z$, that is 
\begin{equation*}
	\sum_{x:xz\in\mathcal{A}}\mu_{xz}-\sum_{y:zy\in\mathcal{A}}\mu_{zy}=q_{z},
\end{equation*}
which we call the \emph{mass balance equation}, and which can be rewritten as
\begin{equation}\label{guitar}
	\nabla^{\intercal}\mu=q.
\end{equation}
The interpretation of these flows will depend on the application of the equilibrium flow problem.  The flows may represent quantities of commodities, volumes of traffic, assignments of objects, matches of individuals, and so on.

\paragraph{Equilibrium flow.} The triple $(q,\mu,p)\in \mathbb R^{\mathcal Z}\times \mathbb R^{\mathcal A}_+\times \mathbb R^{\mathcal Z}$ is an equilibrium flow outcome if it satisfies three conditions.  The first condition is the conservation of the flow given by the mass balance equation~\eqref{guitar}.  The second condition is that there  there is no positive rent on any arc, that is:
\begin{equation*}
	p_x \geq G_{xy}\left(p_y\right)\ ~\forall xy\in\mathcal{A}.
\end{equation*}
Our third condition is that arcs with negative rents carry no flow.   Hence $\mu_{xy}>0\implies  p_x \leq G_{xy}\left(p_y\right)$,
which combines with no-positive-rent requirement to yield
$\mu_{xy}>0\implies p_x = G_{xy}\left(p\right)$. This is a complementary
slackness condition, which can be written
\begin{equation*}
	\sum_{xy\in\mathcal{A}}\mu_{xy} \left( p_x - G_{xy}\left(p_y\right) \right) =0.
\end{equation*}
The interpretation of these rent conditions will again depend on the application.  In some cases, they will be the counterparts of zero-profit conditions in markets with entry, while in other cases they will play the role of incentive constraints.

In summary, we define:

\begin{definition}[{\bf Equilibrium Flow Outcome}] 
	The triple $(q,\mu,p)\in \mathbb R^{\mathcal Z}\times \mathbb R^{\mathcal A}_+\times \mathbb R^{\mathcal Z}$ is an equilibrium flow outcome when the following conditions are met:
	
	\bigskip
	
	(i) $\nabla^{\intercal}\mu=q$
	
	(ii) $p_x \geq G_{xy}\left(p_y\right)\ ~\forall xy\in\mathcal{A}$
	
	(iii) $\sum_{xy\in\mathcal{A}}\mu_{xy} \left( p_x - G_{xy}\left(p_y\right) \right) =0$.
\end{definition}

\noindent The first condition implies $\sum_{z\in\mathcal Z}q_z=0$.  Notice that if $p$ satisfies condition (ii), then setting $q=0$ and $\mu = 0$ ensures that the remaining conditions are satisfied. Indeed, if $(q,\mu,p)$ is a equilibrium flow outcome, then so is $(\lambda q,\lambda \mu,p)$ for any nonnegative scalar $\lambda$.  Hence, there will either be no equilibrium flow outcome (if there is no $p$ satisfying condition (ii)) or there will be multiple equilibrium flows outcomes.  
Figure~\ref{fig:example_equilibrium_flow} presents a simple example of an equilibrium flow outcome. 

\begin{figure}[t]
    \centering
    \includegraphics[width=0.5\textwidth]{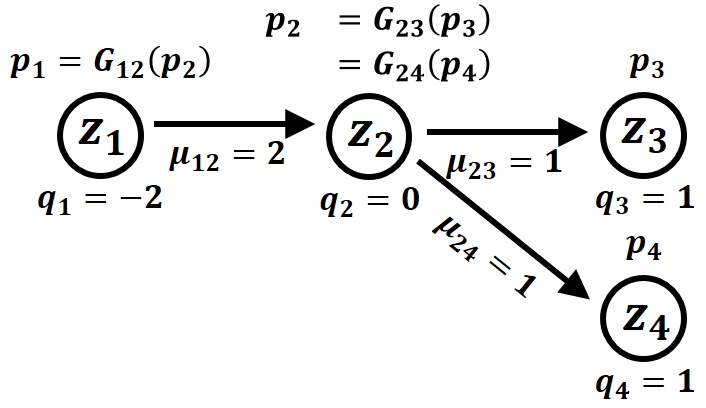}
    \caption{A simple example of equilibrium flow outcome.  The three equilibrium conditions are satisfied as (i) the flow balance condition is satisfied at each node of the network (ii) there is no arbitrage on the network (iii) when there is a trade on an arc ($\mu$ larger than zero on this arc), rent on this arc is nonnegative.}
    \label{fig:example_equilibrium_flow}
\end{figure}

\begin{continueexample}{ex:additive}
		Consider the network $\left(\mathcal{Z},\mathcal{A}%
		\right)$, with the node-arc incidence matrix $\nabla$ and the additive rent function $G_{xy}(p_y) = p_y - c_{xy}$.
		%
		%
		%
		%
		%
		%
		%
		%
		%
		In this additive case, the no-positive-rent condition rewrites as $p_{y}-p_{x}\leq c_{xy}~\forall xy\in\mathcal{A}$. Let  $q$ be a vector of exit flows.  A triple $(q,\mu,p)$ is then an equilibrium flow outcome if
		%
		%
		\begin{eqnarray*}
			(i)&&\nabla^T\mu = q\\
			(ii)&& \left(\nabla p\right)_{a}\leq c_{a}~~~\forall a\in \mathcal A\\
			(iii) &&\mu_{a}>0\implies \left(\nabla p\right)_{a}=c_{a}~~~\forall a\in \mathcal A.
		\end{eqnarray*}
		These conditions are exactly the optimality conditions associated with the
		linear programming problem which consists of minimizing $\mu^\top c$ subject to $\mu \geq 0$ and $\nabla^\top \mu = q$, which has dual $\max_p \left\{ p^\top q : \nabla p \leq c\right\}$. This problem is well-known as the minimum-cost flow problem (see e.g. Ahuja, Magnanti and Orlin \cite{AMandO1993}).  We can interpret the problem as one in which each source node $z$ contains an amount $-q_z$ of material (soil, in the original incarnation of the problem by Monge \cite{Monge1781}), with the objective being to minimize the cost of transporting $q_{\tilde z}$ of this material to each destination $\tilde z$.\hfill\rule{.1in}{.1in}
\end{continueexample}

\begin{example}[Nonadditive shortest path problem]{\em Our setting allows for a natural extension of the well-known (additive) shortest path problem of Bellman \cite{Bellman1958}. The shortest path problem is an 
instance of the minimum-cost flow problem described in the example above, with a single origin or source node $o\in\mathcal{Z}$ and a single destination node $d\in\mathcal{Z}$, and $q_{z}=\mathbf 1_{\left\{ z=d\right\}} -\mathbf 1_{\left\{ z=o\right\} }$. In that case, the additive shortest path problem consists of looking for the path from $o$ to $d$ through the network with the smallest sum of arc costs.
		
We introduce our nonadditive extension by assuming that $p_z$ measures the time of passage at node $z$. Fixing $p_d$, which is the arrival time at the destination node $d$, one seeks the latest departure from  the origin node $o$ consistent with arriving at the destination at no later than time $p_d$.  The innovation allowed by the equilibrium flow formulation is that we can allow the duration of the travel though arc $a$ to depend on the tine the arc is reached, which we can interpret as reflecting the state of traffic, or discrete times of passage of the means of transportation, and so on.

More formally, if a traveler on arc $xy$ aims at arriving at node $y$ at time $p_{y}$, then she must leave node $x$ no later than time $p_{x}=G_{xy}\left( p_{y}\right)$.  To avoid time travel, we assume throughout that the various functions $G_{xy}:\mathbb R\rightarrow \mathbb R$ are increasing and satisfy $G_{xy}\left( p_{y}\right) <p_{y}$. As one can arrive at node $y$ at time $p_y$ by leaving no sooner than $p_x$, the equilibrium should satisfy the condition 
\[
p_{x}\geq G_{xy}\left( p_{y}\right).
\] 
At equilibrium, we have  %
\[
\mu _{xy}>0 ~~\implies ~~p_{x}=G_{xy}\left( p_{y}\right),
\]
which ensures that if arc $xy$ is visited, the time of departure is consistent with the date of arrival. 
%
\hfill\rule{.1in}{.1in}
}
\end{example}

\paragraph{Equilibrium flow correspondence.}  Given an equilibrium flow problem, let the {\em equilibrium flow correspondence} be the correspondence between prices $p$ and quantities $q$ that appear in an equilibrium flow. More formally:

\begin{definition}[Equilibrium flow correspondence]
The equilibrium flow correspondence is the correspondence $\mathtt Q:P\rightrightarrows\mathbb R^{\mathcal Z}$ defined by the fact that for $p\in P$,  $\mathtt Q(p)$ is the set of $q \in \mathbb R^{\mathcal Z}$ such that  there is a flow $\mu$ such that $(q,\mu,p)$ is an equilibrium flow outcome. 
\end{definition}

Appendix \ref{iron} proves:

\begin{theorem}\label{camel}
	The equilibrium flow correspondence $\mathtt Q:\mathbb R^{\mathcal Z}\rightrightarrows\mathbb R^{\mathcal Z}$
	satisfies unified gross substitutes.
\end{theorem}

It is immediate that the correspondence $\mathtt Q$ is nonreversing, since $q\in \mathtt Q(p)$ implies $\sum_{z\in \mathcal Z}q_z=0$.  Hence it follows from Theorem \ref{camel},  the inverse isotonicity Theorem \ref{belgian}, and Corollary \ref{gobble} that:

\begin{corollary}\label{flapper}
	The equilibrium flow correspondence $\mathtt Q(p)$ has totally isotone inverse and the set of equilibrium prices $\mathtt Q^{-1}$ is a sublattice of ${\mathbb R}^{\mathcal Z}$.
\end{corollary}

\noindent Galichon, Samuelson and Vernet \cite{GSV2022} provide an existence result for the equilibrium flow problem, invoking a generalization of Hall's  \cite{Hall1935} conditions and an analogue of Rochet's \cite{Rochet87} cyclical monotonicity condition.

\subsection{Matching with (Fully or Imperfectly) Transferable Utility}

To put our equilibrium flow formulation to work, consider the following one-to-one matching market with either fully or imperfectly transferable utility.  Let  $\mathcal{X}$ be a set of types of workers and $\mathcal{Y}$ a set of types of firms.  There are $n_{x}$ workers of each type $x\in\mathcal{X}$, and $m_{y}$ firms of each type $y\in\mathcal{Y}$.  Assume  the total number of workers and firms is the same, or 
\begin{equation*}
\sum_{x\in\mathcal{X}}n_{x}=\sum_{y\in\mathcal{Y}}m_{y}.
\end{equation*}
A match between worker type $x$ and firm type $y$ is characterized by a wage $w_{xy}$, in which case it gives rise to the utilities $\mathcal{U}_{xy}\left(w_{xy}\right)$ for the worker and $\mathcal{V}_{xy}\left(w_{xy}\right)$ for the firm.     

A matching is a pair $(\mu,w)$, where $w$ is a vector specifying the wage $w_{xy}$ attached to each pair $xy\in \mathcal X\times \mathcal Y$ and $\mu$ is a vector identifying the mass $\mu_{xy}$ of matches between workers of type $x$ and firms of type $y$, for each $xy\in \mathcal X\times \mathcal Y$.  The matching $\left(\mu,w\right)$ is stable if
		\begin{equation}
		\sum_{y\in\mathcal{Y}}\mu_{xy}=n_{x}\text{ and }\sum_{x\in\mathcal{X}%
		}\mu_{xy}=m_{y}
		\end{equation}
		and
		\begin{equation}
		\mu_{xy}>0\text{ implies }y\in\arg\max_{y\in\mathcal{Y}}\mathcal{U}%
		_{xy}\left(w_{xy}\right)\text{ and }x\in\arg\max_{x\in\mathcal{X}}\mathcal{V}%
		_{xy}\left(w_{xy}\right).\label{two-maximization-pbs}
		\end{equation}
The first condition provides the feasibility condition that the number of each type of worker and firm that is matched equals the number present in the market, while the second provides the stability condition that no worker and firm can improve their utilities by matching with one another at an appropriate wage.

Now let $\mathtt Q$ be a correspondence identifying, for each specification of utilities, the configurations of workers and firms for which there exists a stable match exhibiting those utilities. More formally, let $\mathcal{Z}=\mathcal{X}\cup \mathcal{Y}$ and fix the utility vector $u:\mathcal Z\rightarrow \mathbb R$.  Let  $p\in \mathbb{R}^{\mathcal{Z}}$ be defined by  $p_{z}=u_{z}\mathbf 1_{\left\{ z\in \mathcal{X}\right\}} -v_{z}\mathbf1_{\left\{ z\in \mathcal{Y}%
\right\}} $, and $q_{z}=-n_{z}\mathbf1_{\left\{ z\in \mathcal{X}\right\}}
+m_{z}\mathbf1_{\left\{ z\in \mathcal{Y}\right\}} $.%
\footnote{Notice that the quantity vector $q$ identifies the negative of the number of each type of workers and the number of each type of firm.  This allows us to represent a match between a work and a firm as a flow along the arc connecting the node containing that type of worker to the node representing the firm.  Similarly, the payoff vector $p$ identifies the payoffs of workers and the negative of the payoffs of firms, so that an increase in the payoff of a firm at node $y$ corresponds to a smaller utility requirement from the firm at that node, making it more attractive for workers to traverse the arc terminating at the node, which corresponds to the increasingness of $G_{xy}(p_y)$. }
The stability problem then 
reformulates as%
\[
q\in \mathtt{Q}\left( p\right), 
\]%
where $\mathtt{Q}(p)$ is the set of 
$q\in \mathbb{R}^{\mathcal{Z}}$ such that there exists  $w\in \mathbb{R}^{\mathcal{X}\times 
\mathcal{Y}}$ and $\mu \in \mathbb{R}_{+}^{\mathcal{X}\times \mathcal{Y}}$ with $\sum_{y\in \mathcal{Y}}\mu _{xy}=-q_{x}$ for all $x\in \mathcal{X\,\ }\text{%
and }\sum_{x\in \mathcal{X}}\mu _{xy}=q_{y}$ for all $y\in \mathcal{Y}$
and also \eqref{two-maximization-pbs} is satisfied with $p_z = \max_{y\in\mathcal{Y}}\mathcal{U}%
_{xy}\left(w_{xy}\right)$ for $z\in \mathcal X$ and $-p_z = \max_{x\in\mathcal{X}}\mathcal{V}%
_{xy}\left(w_{xy}\right)$ for $z\in \mathcal Y$.

\begin{figure}[t]
    \centering
    \includegraphics[width=0.35\textwidth]{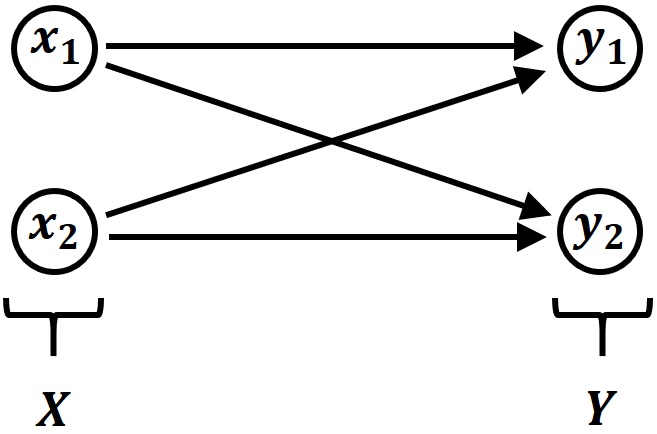}
    \caption{Network associated with a bipartite matching problem without singles.}
    \label{fig:matching-without-singles}
\end{figure}

As a immediate result of the equilibrium flow formulation illustrated in Figure~\ref{fig:matching-without-singles}, it follows from Theorem~\ref{camel} that:

\begin{theorem}\label{thm:matching-is-ugs}
The correspondence that maps the vector of payoffs $p$  to the vector of populations $q$ is an M0-correspondence.  
\end{theorem}

\noindent As noted in Galichon \cite[Section 4.3.2]{galichon2021unreasonable},  it may be paradoxical to see substituability arise in matching problems, given that the latter are  meant to capture complementarity. However,  the ``change-of-sign'' technique applied here, consisting of adding a negative sign in front of $v_y$ and in front of $n_x$, allows one to turn a problem with complementarity into a problem with substitutes. This exploits the bipartite nature of the problem.\footnote{The same phenomenon explains why convex submodular functions are dual to convex supermodular functions in dimension two, but not above.}

\begin{remark}[Equilibrium utilities constitute a lattice]\label{remark:unmatched}{\em 
Appendix \ref{chatter} generalizes  this matching problem to accommodate unmatched agents, showing it is equivalent to an equilibrium flow problem.   Theorem \ref{camel} again applies, and in each case Corollary \ref{flapper} then implies that  $\mathtt{Q}$ is an M0-correspondence.  Hence, given a specification $q$ of buyers and sellers, the  set of equilibrium utilities consistent with a stable match constitutes a lattice, providing an alternative route to the lattice result of Demange and Gale \cite[Lemma 2 and Property 2]{demange1985strategy}.\hfill\rule{.1in}{.1in}}
\end{remark}

\begin{remark}[Comparative Statics]{\em 
Once again, the inverse isotonicity of the equilibrium correspondence gives comparative static results.  For example, as the number of firms increases, the set of equilibrium payoffs of firms decreases (in the strong set order) while set set of equilibrium payoffs of workers increases.  An increase in the number of workers has the reverse effects. This result is proven in \cite{demange1986multi}. \hfill\rule{.1in}{.1in}}
\end{remark}

\subsection{Matching Without Transfers}

We now consider bipartite matching with non-transferable utility. We borrow
the traditional language of \textquotedblleft men\textquotedblright\ and
\textquotedblleft women\textquotedblright\ and heterosexual unions. Let $%
\mathcal{X}$ the set of men and $\mathcal{Y}$ be the set of women.  Let  $\mathcal{X}_{0}=\mathcal{X}\cup \left\{ 0\right\} $ be the set of
marital options available to women, where $0$ is being unmatched, and
similarly let $\mathcal{Y}_{0}=\mathcal{Y}\cup \left\{ 0\right\} $ be the set
of marital options available to men.

A match between a man $x\in \mathcal{X}$ and a woman $y\in \mathcal{Y}$
generates utility $\alpha _{xy}$ for the man and $\gamma _{xy}$ for the
woman; if $x$ remains unmatched he receives $\alpha _{x0}$, while if $y$ remains
unmatched, she receives $\gamma _{0y}$. As is common in the literature, we impose strict preferences:

\begin{assumption}\label{strict}
\textbf{Strict preferences}. $\alpha _{xy}\neq \alpha _{xy^{\prime }}$ for
every $x\in \mathcal{X}$ and $y\neq y^{\prime }\in \mathcal{Y}_{0}$, and $%
\gamma _{xy}\neq \gamma _{x^{\prime }y}$ for every $y\in \mathcal{Y}$ and $%
x\neq x^{\prime }\in \mathcal{X}_{0}$.
\end{assumption}

\paragraph{Matchings, feasible matchings.} A matching $\mu =\left( \mu _{xy},\mu _{x0},\mu _{0y}\right) $ is such that $%
\mu _{xy}=1$ if man $x$ and woman $y$ are matched and equals $0$ otherwise, $\mu
_{x0}=1$ if man $x$ remains unmatched and $0$ otherwise, and $\mu _{0y}=1$
if woman $y$ remains unmatched, and $0$ otherwise.  A {\em feasible} matching is such that each individual has at most one partner,
that is%
\begin{equation*}
\sum_{y\in \mathcal{Y}_{0}}\mu _{xy}=1\text{ and }\sum_{x\in \mathcal{X}%
_{0}}\mu _{xy}=1.
\end{equation*}
Given a feasible matching $\mu $, one introduces 
\begin{equation}
u_{x}^{\mu }=\sum_{y\in \mathcal{Y}_{0}}\mu _{xy}\alpha _{xy}\text{ and }%
v_{y}^{\mu }=\sum_{x\in \mathcal{X}_{0}}\mu _{xy}\gamma _{xy},
\label{def-umu-and-vmu}
\end{equation}
which are the utilities obtained by each individual under this matching.

\paragraph{Stable matchings.} A \emph{stable matching} is a feasible matching such that there is no pair $xy\in 
\mathcal{X}\times \mathcal{Y}$ for which $\alpha _{xy}>u_{x}^{\mu }$ and $%
\gamma _{xy}>v_{y}^{\mu }$ and such that $u_{x}\geq \alpha _{x0}$ for all $%
x\in \mathcal{X}$ and $v_{y}\geq \gamma _{0y}$ for all $y\in \mathcal{Y}$.

Define $\mathtt Q:\mathbb{R}^{\mathcal{Y}}\rightarrow \mathbb{Z}^{\mathcal{Y}}$ by%
\begin{equation*}
\mathtt Q_{y}\left( v\right) :=1-\sum_{x\in \mathcal{X}}\mathbf 1_{\left\{ y\in \arg
\max_{y\in \mathcal{Y}}\left\{ \alpha _{xy}:\gamma _{xy}\geq v_{y},\alpha
_{x0}\right\} \right\}} -\mathbf1_{\left\{ \gamma _{0y}\geq v_{y}\right\}},
\end{equation*}
which can be interpreted as the ``excess supply'' of type $y$, namely the ``supply'' of type $y$ (which is one) when this agent is in a relation that secures him a utility $v_y$, minus the  the number of  agents $x\in \mathcal X_0$ on the other side of the market who are willing to match with type $y$ (given $y$'s utility requirement $v_y$). The following theorem (with proof in Appendix \ref{proof_th_ntu_matching}) expresses stable matchings in terms of the function $\mathtt Q$.

\begin{theorem}\label{th_ntu_matching}
Let assumption \ref{strict} hold. Then:

(i) If $\mu $ is a stable matching, then $\mathtt Q\left( v^{\mu }\right) =0$, with $v^{\mu }$  defined in (\ref{def-umu-and-vmu}).

(ii) Conversely, if $\mathtt Q\left( v\right) =0$, then define 
\begin{equation*}
\begin{array}{l}
\mu _{xy}=\mathbf 1_{\left\{ y\in \arg \max_{y\in \mathcal{Y}}\left\{ \alpha
_{xy}:\gamma _{xy}\geq v_{y},\alpha _{x0}\right\} \right\}} \\ 
\mu _{x0}=\mathbf 1_{\left\{ 0\in \arg \max_{y\in \mathcal{Y}}\left\{ \alpha
_{xy}:\gamma _{xy}\geq v_{y},\alpha _{x0}\right\} \right\}} \\ 
\mu _{0y}=\mathbf 1_{\left\{ \gamma _{0y}\geq v_{y}\right\}},%
\end{array}%
\end{equation*}%
and we have that $\mu =\left( \mu _{xy},\mu _{x0},\mu _{0y}\right) $ is a
stable matching.
\end{theorem}

\noindent This result is useful because we can show that $\mathtt Q$ is an M0-correspondence.  Appendix \ref{proof_th_ntu_matching_m0} proves the following result.  

\begin{theorem}\label{th_ntu_matching_m0}
The map $\mathtt Q$ defined above defines a (point-valued) M0-correspondence.
\end{theorem}

\noindent  This gives us an alternative route to the lattice structure of stable matchings.

\subsection{Hedonic Pricing}

The simplest models of a competitive economy assume that each of a finite number of goods is perfectly divisible and perfectly homogeneous .  The hedonic pricing model, introduced by Rosen \cite{Rosen1974} and developed by  Ekeland, Heckman and Nesheim \cite{EHN2004}) and  Chiappori, McCann and Nesheim \cite{CMandN2010} (among others), examines the opposite extreme, examining an economy filled with indivisible, idiosyncratic goods.    To reduce the dimensionality of the prices in the latter case, one typically assumes the goods can be described by the extent to which they exhibit certain characteristics, with prices determined by these characteristics, thus giving rise to the term hedonic pricing.  

Chiappori, McCann and Nesheim \cite{CMandN2010} show the the hedonic pricing problem with quasilinear utilities is equivalent to a matching problem with transferable utility, which is in turn equivalent to an optimal transport problem.  They then draw on familiar results to establish that the optimal transport problem has an equilibrium, and hence so do the the associated matching and hedonic pricing problems.  

We examine the hedonic pricing problem without assuming that utilities are quasilinear.       Chiappori, McCann and Nesheim \cite{CMandN2010} invoke a twist condition to establish the uniqueness for equilibrium for the quasilinear case.  Our counterpart of this result is to establish that for any equilibrium allocation, the set of utilities and prices supporting this equilibrium allocation is a lattice.  We avoid  a host of technical difficulties by working with finite sets of agents and goods.

\paragraph{Consumers and producers.} The basic elements of the model are a set $\mathcal{X}$ of types of producers and a set $\mathcal{Y}$ of types of consumers.  There are $n_x$ producers of each type $x\in \mathcal{X}$ and $m_y$ consumers of each type $y\in \mathcal{Y}$.  We do not require that the total number of producers $\sum_{x\in
	\mathcal{X}}n_{x}$ equal the total number of consumers  $\sum_{y\in \mathcal{Y}}m_{y}$.

\paragraph{Qualities.} There is a finite set $\mathcal{W}$ of qualities, also sometimes referred to as contracts or characteristics.  Each  producer must choose to produce one of the qualities  in $\mathcal{W}$, or to remain inactive.  Each consumer must choose to consume one quality  in $\mathcal{W}$ or remain inactive.  We can interpret the set $\mathcal{W}$ as the set of possible (vectors of) characteristics of goods that determine prices.

\paragraph{Hedonic prices.} Let $p:\mathcal{W}\rightarrow \mathbb R$ be a price vector assigning prices to qualities, with 
$p_w$ denoting the price of quality $w$.   A producer of type $x$ who produces a quality $w$ that bears price $p_w$ earns the profit $\pi_{xw}(p_w)$.  A consumer of type $y$ who consumes a quality $w$ bearing price $p_w$ reaps surplus $s_{yw}(p_w)$.  

\paragraph{Indirect utilities.} Given the price function $p$, a producer of type $x$ solves
\[
\max_{w\in \mathcal{W}}\left\{ \pi _{xw}\left( p_{w}\right) ,0\right\}=:u_x(p)
\]
and a consumer solves of type $y$ solves
\begin{equation*}
\max_{w\in \mathcal{W}}\left\{ s_{yw}\left( p_{w}\right) ,0\right\}=:v_y(p).
\end{equation*}

\paragraph{Supply and demand allocations.} Introduce $\mu _{xw}$ as the number of producers of type $x$ producing
quality $w$, and $\mu _{wy}$ as the number of consumers of type $y$
consuming quality $w$. Complete this notation by introducing $\mu _{x0}$ and
$\mu _{0y}$ as respectively the number of consumers of type $x$ and
producers of type $y$ opting out. Necessarily if $\mu_{xw}>0$, it must hold that $\pi_{xw}(p_w) = u_x(p)$, that is, $w$ is produced by $x$ only if producing $w$ is optimal for $x$. Similar considerations apply for the demand allocation.

\paragraph{Hedonic pricing equilibrium.} Given the specifications $\{n_x\}_{x\in \mathcal{X}}$ and $\{n_y\}_{y\in \mathcal{Y}}$ of producers and consumers, as well as the functions $\pi:\mathcal{X}\times \mathcal{W} \times \mathbb R\rightarrow \mathbb R$ and 
$s:\mathcal{Y} \times \mathcal{W} \times \mathbb R\rightarrow \mathbb R$, an equilibrium determines the price function $p$ and the specification of which good (if any) is produced by each producer and which good (if any) is consumed by each consumer.  In the process, the equilibrium determines the payoff $u_x$ of each producer $x$ and $v_y$ of each consumer $y$, as well as the quantity $q_w$ of each quality $w$ produced.  

\vspace{.1in}

\begin{definition}\label{book}
A price vector $p$ and allocation $\mu$ are a {\em hedonic pricing equilibrium} if the attendant  utilities of producers and consumers $u_x$ and $v_y$, the prices of the qualities $p_w$, the numbers of producers and consumers $n_x$ and $m_y$, the excess supply of qualities $q_w$, the production flows $\mu_{xw}$, and the consumption flows $\mu_{wy}$
are related by the following relations: 
\begin{itemize}

\item[(i)] supply and demand allocations are feasible:  
\begin{equation}
\sum_{w\in \mathcal{W}}\mu _{xw}+\mu _{x0} = n_{x} \text{ and } \sum_{w\in \mathcal{W}}\mu _{wy}+\mu _{0y} = m_{y},\label{hed-1}
\end{equation}

\item[(ii)] markets balance holds:
\begin{equation}\label{hed-2}
q_{w} = \sum_{x\in \mathcal{X}}\mu _{xw} - \sum_{y\in \mathcal{Y}}\mu _{wy},
\end{equation}

\item[(iii)] rents are nonpositive and agents maximize:
\begin{equation}\label{hed-3}
\left\{ 
\begin{array}{lll}
u_{x}\geq \pi _{xw}\left( p_{w}\right)  & \text{ with equality if } & \mu
_{xw}>0 \\ 
v_{y}\geq 0 & \text{ with equality if } & \mu _{0y}>0 \\ 
v_{y}\geq s_{yw}\left( p_{w}\right)  & \text{ with equality if } & \mu
_{wy}>0 \\ 
u_{x}\geq 0 & \text{ with equality if } & \mu _{x0}>0%
\end{array}%
\right. 
\end{equation}
\end{itemize}
\end{definition}	

It is shown in Appendix~\ref{app:hedonic} that:
\begin{theorem}\label{thm:ugs-for-hedonic} The correspondence that associates the vector  $(u,p,-v)$ to the set of vectors $(-n,q,m)$ such that equations~\eqref{hed-1},~\eqref{hed-2} and~\eqref{hed-3} are satisfied for some allocation $\mu \geq 0$ is an M0-correspondence.  
\end{theorem}

\noindent The proof of this result is based on a reformulation as an equilibrium flow problem, and the application of Theorem \ref{camel}. This equilibrium therefore  exhibits the lattice structure established in  Corollary \ref{gobble}, and we once again obtain comparative static results.. We can again turn to Galichon, Samuelson and Vernet \cite{GSV2022} for an existence result.

\section{Conclusion}

The concept of weak gross substitutes plays a prominent role in economic theory.  
We view the concept of unified gross substitutes as the natural generalization of weak gross substitutes to correspondences.  It connects to the literature in multiple points, generalizing some results and unifying others.

The concept of unified gross substitutes allows one to derive the inverse isotonicity and lattice-valued-inverse properties of general correspondences, which in turn gives rise to comparative static results.  This provides a tool that should be useful in two directions.  First, it should be useful in extending familiar results developed under the assumption of quasilinearity to more general settings.  Second, it should allow research to address a broader class of problems.  In particular, there is great potential for formulating a variety of problems as special cases of the equilibrium flow problem, allowing immediate application of the implications of unified gross substitutes.

\appendix 

\section{Appendix}\label{app:proofs}

\subsection{Aggregation Preserves Unified Gross Substitutes\label{woodpecker}}
We prove one of the four implications in the definition of unified gross
substitutes, noting that the others follow along similar lines. Fix $p$ and $%
p^{\prime }$ and let
\begin{eqnarray*}
	q=\lambda q^1+\mu q^2, &&q^1\in \mathtt Q^{1}(p),~~~~q^2\in \mathtt Q^{2}(p) \\
	q^{\prime }=\lambda q'^1+\mu q'^2 &&q'^1\in
	\mathtt Q^{1}(p^{\prime }),~~~~q'^2\in \mathtt Q^{2}(p^{\prime }).
\end{eqnarray*}%
Because $\mathtt Q^{1}$ and $\mathtt Q^{2}$ satisfy unified gross substitutes, there
exist $q^{\vee 1}\in \mathtt Q^{1}(p\vee p^{\prime })$ and $q^{\vee 2}\in
\mathtt Q^{2}(p\vee p^{\prime })$ with
\begin{eqnarray*}
	p_{z}\leq p_{z}^{\prime } &\implies &q_{z}^{\vee 1}\leq q_{z}'^1
	\\
	p_{z}\leq p_{z}^{\prime } &\implies &q_{z}^{\vee 2}\leq q_{z}'^2.
\end{eqnarray*}%
But then we have
\begin{equation*}
	p_{z}\leq p_{z}^{\prime }\implies \lambda q_{z}^{\vee 1}+\mu q_{z}^{\vee
		2}\leq \lambda q_{z}'^1+\mu q_{z}'^2,
\end{equation*}%
giving the condition for unified gross substitutes.

\subsection{The Sum of Two M0-correspondences is Not Always an M0-correspondence}\label{par:sum-m-corresp}

Let $\mathtt Q\left( p\right) =\left\{ Mp\right\} $ and $\mathtt Q'\left( p\right) =\left\{
M^{\top }p\right\} $ be two point-valued correspondences respectively
associated with matrices $M$ and $M^{\top }$, where%
\begin{equation*}
M=%
\begin{pmatrix}
~~5 & -1 \\ 
-4 & ~~1
\end{pmatrix}~~.
\end{equation*}
Because $M$ and $M^{\top}$ have positive entries on the diagonal and negative entries off the diagonal, the functions $\mathtt Q$ and $\mathtt Q'$ both satisfy weak gross substitutes and hence unified gross substitutes (cf. Property 	 \ref{prop:wgs-implies-ugs}).  
We have 
\begin{equation*}
M^{-1}=%
\begin{pmatrix}
1 & 1 \\ 
4 & 5%
\end{pmatrix}~~.
\end{equation*}%
The inverses of $M$ and $M'$ are both composed of positive entries, so $\mathtt Q$ and $\mathtt Q'$ are both inverse isotone and hence nonreversing (Property \ref{monica}) and hence are M0-correspondences.    However, 
\begin{equation*}
\left( M+M^{\top }\right) ^{-1}=\frac{1}{10}%
\begin{pmatrix}
-4 & -10 \\ 
-10 & ~~20%
\end{pmatrix}
\end{equation*}
is not entrywise positive.  Hence, the inverse of $\mathtt Q+\mathtt Q^{\prime }$ is not isotone and so $\mathtt Q+\mathtt Q^{\prime }$ is not an
M0-correspondence (Theorem \ref{belgian}).  Unified gross substitutes is satisfied (see Property~\ref{pty:agg-pres-ugs}) but $\mathtt Q+\mathtt Q'$ is not nonreversing.

\subsection{Kelso and Crawford's Substitutes Does Not Imply Unified Gross Substitutes}
\label{par:kc-example}

We show that unified gross substitutes  is strictly stronger than Kelso and Crawford's notion of gross substitutes. 

Let $N=4$ and consider the price vectors
		\begin{eqnarray*}
			p  = (1,1,2,2)&&p\vee p^{\prime}=(2,2,2,2) \\
			p\wedge p^{\prime}=(1,1,1,1)&&~~~~~p^{\prime} = (2,2,1,1).
		\end{eqnarray*}
Let the supply correspondence be
		\begin{eqnarray*}
			\mathtt Q(p) = \left\{\left(1,0,1,1\right),\left(0,1,1,0\right)\right\}			&&~~~~~~~~\mathtt Q(p\vee p^{\prime})=\{
			\left(1,1,1,0\right),\left(0,1,1,1\right)\} \\
			\mathtt Q(p\wedge p^{\prime})= \left\{
			\left(1,0,0,0\right),\left(0,0,0,1\right)\right\}
			&&~~~~~~~~~~~~~\mathtt Q(p') =\left\{ \left(1,1,0,1\right),\left(0,1,1,0\right)\right\}.
		\end{eqnarray*}
		Then Kelso and Crawford's condition holds, but unified gross substitutes fails. The difficulty is that unified gross substitutes requires the
		allocations $q^{\vee}$ and $q^{\wedge}$ that appear in the two parts of Definition \ref{duck}   to be identical, while Kelso and Crawford's definition has no counterpart of this requirement. 

Note also that $\mathtt Q$ is nonreversing, but Kelso and Crawford's condition is not strong enough to ensure  total inverse isotonicity, as $\left( 0,1,1,0\right) \in Q\left( p\right)$ and $\left( 0,1,1,0\right) \in Q\left( p'\right)$, yet $\left( 0,1,1,0\right) \notin Q\left( p\wedge p^{\prime }\right)$.

\subsection{Unified Gross Substitutes and Polterovich and Spivak's Property Are Independent}\label{par:polterovich-comparison}

The following example shows that Polterovich and Spivak's definition of gross substitutes does not imply unified gross substitutes.  

Let $N=4$ and
\begin{eqnarray*}
	p  =(1,1,2,2) ~~&&~~p\vee p^{\prime}=(2,2,2,2)\\
	p\wedge p^{\prime}=(1,1,1,1)~~	&&~~~~~~~p^{\prime} = (2,2,1,1),
\end{eqnarray*}
and let the supply correspondence be
\begin{eqnarray*}
	\mathtt Q(p)  =(0,1,2,1) ~~&&~~\mathtt Q(p\vee p^{\prime})=(0,2,2,0)\\
	\mathtt Q(p\wedge p^{\prime})=(1,0,2,1)~~	&&~~~~~~~\mathtt Q(p^{\prime}) = (1,2,2,0).
\end{eqnarray*}
Then Polterovich and Spivak's \cite{Posp1983} notion of gross substitutability is satisfied.   For each of the four relevant increasing-price comparisons ($p\wedge p'$ to $p$, $p\wedge p'$ to $p'$, $p$ to $p\vee p'$ and $p'$ to $p\vee p'$), among the prices that remain constant, there is at least one dimension on which the allocation decreases.  However, unified gross substitutes fails.  We have $p_2<p_2'$, but $q_2\le q_2^{\wedge}$ fails.  Similarly, 
$p_4'<p_4$, but $q_4^{\vee}\le q_4'$ fails.

We next show that one can have unified gross substitutes without having Peltorovich and Spivak's notion of gross substitutability.
Let $N=2$ and
\[
p = (1,1)~~~~~~p'=(1,2)
\]
and let the supply correspondence be
\[
\mathtt Q(p) =\{(1,0),(3,0)\}~~~~~~~~\mathtt Q(p') = \{(0,0),(2,0)\}.
\]
Then Polterovich and Spivak's \cite{Posp1983} notion of gross substitutes fails---as we move from $p$ to $p'$, the price of good 2 increases while that of good 1 remains constant, but some allocations in $\mathtt Q(p')$ supply more good 1 than do some allocations in $\mathtt Q(p)$.  Unified gross substitutes is satisfied---for every allocation in $\mathtt Q(p')$, there is an allocation in $\mathtt Q(p)$ that supplies less good 1, and for every allocation in $\mathtt Q(p)$, there is an allocation in $\mathtt Q(p')$ that supplies more good 1, which in this simple case exhausts the implications of unified gross substitutes.

\subsection{Outside Goods and Weighted Monotonicity}\label{horns}

Given a correspondence $\mathtt Q:P\rightrightarrows Q$, we can introduce a fictitious good $0$ with price $p_{0}$
and constants $k\in\mathbb R^N_{++}$, and define the {\em extended correspondence} $\tilde {\mathtt Q}:P\times P_0\rightrightarrows Q\times \mathbb R$ by letting 
\begin{equation*}
	\tilde{\mathtt Q}\left( p,p_{0}\right) =\left\{\left( q,q_{0}\right) \in
	Q\times \mathbb{R}:q\in E\left( p\right)
	,q_{0}=p_{0}-\sum_{z=1}^Nk_zq_{z}\right\}.
\end{equation*}%
The obvious candidates for extended correspondences will set $P_0 = \{1\}$ and each of the constants $k_z=1$.  

We then have:

\begin{lemma}
	If the extended correspondence $\tilde {\mathtt Q}$ satisfies unified gross substitutes, then the correspondence $\mathtt Q$ satisfies unified gross substitutes and weighted monotonicity.
\end{lemma}

\paragraph{Proof.}  Let $\tilde {\mathtt Q}$, derived from $\mathtt Q$ via the constants $k\in\mathbb R^N_{++}$, satisfy unified gross substitutes.   The conditions for $\mathtt Q$ to satisfy unified gross substitutes are a subset of the conditions for $\tilde {\mathtt Q}$ to do so, and hence it is immediate that $\mathtt Q$ satisfies unified gross substitutes.  Next, fix the price vectors $(p,p_0)$ and $(p',p'_0)$.  Then applying unified gross substitutes to $\tilde {\mathtt Q}$, we have
\[
p_0<p_0' \implies q_0 \le q_0^{\wedge}~~~{\rm and}~~q_0^{\vee}\le q_0'.
\]
We then note that $p_0=p_0^{\wedge}$ and $p'_0= p_0^{\vee}$, and hence 
\begin{eqnarray*}
	p_0 -q_0 \ge p_0^{\wedge}-q_0^{\wedge}\\
	p'_0 -q'_0 \le p_0^{\vee}-q_0^{\vee}.
\end{eqnarray*}
Applying the definition of the extended correspondence gives
\begin{eqnarray*}
	p_0 -q_0 \ge p_0^{\wedge}-q_0^{\wedge} &\implies &\sum_{z=1}^Nk_zq_{z}\geq \sum_{z=1}^Nk_zq_{z}^{\wedge }\\
	p'_0 -q'_0 \le p_0^{\vee}-q_0^{\vee} &\implies &
	\sum_{z=1}^Nk_zq_{z}^{\vee }\geq \sum_{z=1}^Nk_zq_{z}^{\prime },
\end{eqnarray*}
giving the weighted monotonicity of $\mathtt Q$.
\hfill\rule{.1in}{.1in}





\subsection{Connections with Monotone Comparative Statics}\label{buffalo}

We first present an example in which the correspondence $\mathtt Q$ fails unified gross substitutes, but satisfies \eqref{string} for a function $g(p,q)$ satisfying \eqref{spinner} and hence Topkis' assumptions, and hence $\mathtt Q^{-1}$ is isotone.  

Let $P=S=\mathbb R^3$ and let $\mathtt Q(p) = \{ Yp \}$ where $Y$ is the matrix
\[
\frac{1}{28}
\left[
\begin{array}{ccc}
	63&-28&-28\\
	-28&16&12\\
	-28&12&16
\end{array}
\right].
\]
It is easy to see that the correspondence $\mathtt Q$ is a function that fails weak gross substitutes (because $Y$ has some positive off-diagonal terms), and hence fails 	unified gross substitutes (cf. Property~\ref{prop:wgs-implies-ugs}). The inverse is the point-valued correspondence $\mathtt Q^{-1}= \{ Xq \}$ where $X$ is the inverse of $Y$, given by
\[
\left[
\begin{array}{ccc}
	4&4&4\\
	4&8&1\\
	4&1&8
\end{array}
\right].
\]
Since $X$ has strictly positive terms, it is  immediate that $\mathtt Q^{-1}$ is isotone.  The function $\mathtt Q^{-1}$ is the solution to the maximization problem%
\footnote{In particular, the first-order conditions for this maximization problem are $Xq=p$, while the sufficient second-order conditions are clearly met.
}
\[
\max_{p\in \mathbb R^3}g(p,q) = \max_{p\in \mathbb R^3}  pXq-\frac12p_1^2-\frac12p_2^2-\frac12p_3^2.
\]
The function $g(p,q) =   pXq-\frac12p_1^2-\frac12p_2^2-\frac12p_3^2$ is strictly concave and satisfies Topkis' assumptions given in \eqref{spinner}.

We next present an equilibrium correspondence that satisfies unified gross substitutes and aggregate monotonicity, and hence has an isotone inverse, but for which the corresponding function $g$ satisfying \eqref{string} fails the formulation of Milgrom and Shannon's conditions given by \eqref{rooster} (and hence, fails those of Topkis).  

Let $P = Q = \mathbb R^2_+$ and let the correspondence $\mathtt Q$ be given by
\begin{eqnarray*}
	q_1(p) &=& p_2\\
	q_2(p) &=& p_1.
\end{eqnarray*}
Clearly $\mathtt Q$ is point-valued, and satisfies unified gross substitutes and aggregate monotonicity.  Let $g(p,q)=1$ if $q=(p_2,p_1)$ and $g(p,q)=0$ otherwise, and note that \eqref{string} holds.  Consider the second implication in \eqref{rooster}, namely
\[
g\left(p\vee
p^{\prime},q\right)-g\left(p^{\prime},q\right)\leq0%
\implies g\left(p,q'\right)-g\left(p\wedge p^{\prime},q'\right)\leq0.
\]
Letting
\begin{eqnarray*}
	p = (3,3)&~~~~~~&q = (5,2)\\
	p' = (2,5)&&q' = (4,1)
\end{eqnarray*}
gives us a violation.

\subsection{Proof of Theorem~\ref{spice}}\label{purple}
We present a proof of Theorem~\ref{spice} that establishes the equivalence between unified gross substitute of $\partial c^*$ and the submodularity of $c$.
In order to do this, we prove a series of lemmas regarding convex functions and how to characterize their submodularity.  Let $f$ be a convex function.  Our interpretation will be that $f$ is an indirect profit function $f=c^*$ associated with a cost function $c(q)$, but we will not use that interpretation in the lemmas.

The first result is a well-known result in convex analysis (Theorem 23.4 in Rockafellar \cite{Rocket1997}), recalled here for convenience, which essentially asserts that the support function of the subdifferential of a convex function coincides with the directional derivatives.

\begin{lemma}\label{lem:char-subdiff}
Let $f: \mathbb R^N \to \mathbb R$ be a convex function. We have%
	\begin{equation*}
		\partial f\left( p\right) =\left\{ q\in \mathbb R^{N}:q^{\top }b\leq \frac{d%
		}{dt}f\left( p+tb\right) |_{0^{+}},\forall b\in \mathbb \mathbb \mathbb R^{N}\right\}.
	\end{equation*}    
\end{lemma}

\begin{proof}[Proof of Lemma~\ref{lem:char-subdiff}]
By definition, one has that $\partial f\left( p\right)$ is the set of $q\in \mathbb R^N$ such that the quantity $q^{\top }p^{\prime }-f\left( p^{\prime }\right)$  is maximal for  $p^\prime = p$. Because $f$ is convex, the function $p^{\prime }\rightarrow q^{\top }p^{\prime
}-f\left( p^{\prime }\right) $ is maximal at $p$ if and only if all
the directional derivatives at $p$ are nonpositive, thus for all $b\in \mathbb R^{N}$,
\[
	\frac{d}{dt}\left\{ q^{\top }\left( p+tb\right) -f\left( p+tb\right)
	\right\} |_{0^{+}}\leq 0,
	\]
or equivalently,
\[q^{\top }b\leq \frac{d}{dt}c^{\ast
	}\left( p+tb\right) |_{0^{+}}.
\]
\end{proof}
	
	\bigskip

For $X\subseteq \mathbb R^N $, define $\tilde{X}$ as the set of vectors of $\mathbb R^N$ that are dominated by some vector in $X$, or more formally:
\begin{equation}\label{eq:Ctilde}
    \tilde{X}=\left\{ \tilde{q}\in \mathbb R^{N}:\exists q\in
    X~s.t.~\tilde{q}\leq q\right\} ,
\end{equation} 
and define the \emph{support function} of $X$ as
\begin{equation}\label{eq:support-function}
h_{X}\left( b\right) =\sup_{q\in X} \left\{ q^{\top }b:\right\}.
\end{equation}
Let $cch\left( X\right) $ be
the convex closure of $X$, which is the closure of the convex hull of $X$. It is well-known that $cch\left( X\right) $  is the set of elements $x\in \mathbb{R}^N$ such that $x^\top b \leq h_X(b)$ for all $b\in \mathbb{R}^N$.

The next result states that the support function of $\tilde C$ is the support function of $C$ whose domain has been restricted to nonnegative coordinates. 

\begin{lemma}\label{tympani}
	If $C$ is a closed convex set of $\mathbb R^N$, then $\tilde{C}$ as defined in expression~\eqref{eq:Ctilde} can be expressed as
	$$ \tilde{C}=\left\{ q:q^{\top }b\leq h_C\left( b\right) ,\forall b\in
	\mathbb R_{+}^{N}\right\}$$
	where $h_C$ has been defined in \eqref{eq:support-function}. 
\end{lemma}

\begin{proof}[Proof of Lemma~\ref{tympani}]
	By the supporting hyperplane theorem, for any convex set $C$ and any boundary point $q_0 = \arg \max_{q\in C}q^{\top }b$ of the boundary of $C$ there exists a supporting hyperplane for $C$ at $q_0$. Therefore one has%
	\begin{equation*}
		C=\left\{ q:q^{\top }b\leq h\left( b\right) ,\forall b\in \mathbb R^{N}\right\}.
	\end{equation*}
	Further, note that $\hat{C}=\left\{ \hat{q}:\hat{q}^{\top }b\leq \hat{h}%
	\left( b\right) ,\forall b\in \mathbb R^{N}\right\} $, where $\hat{h}\left( b\right)
	=h\left( b\right) $ if $b\in \mathbb R_{+}^{N}$ and  $\hat{h}\left( b\right)
	=+\infty $ otherwise, and thus 
	\begin{equation*}
		\hat{h}\left( b\right) =\max_{\hat{q}\in \hat{C}}\hat{q}^{\top }b.
	\end{equation*}
	Now compute $\tilde{h}\left( b\right) =\max_{\tilde{q}\in \tilde{C}}\tilde{q}%
	^{\top }b$. One has $\tilde{C}=\left\{ q-\delta :q\in C,\delta \in
	\mathbb R_{+}^{N}\right\} $, and so%
	\begin{equation*}
		\tilde{h}\left( b\right) =\max_{q\in C}\max_{\delta \geq 0}\left( q-\delta
		\right) ^{\top }b.
	\end{equation*}%
	Thus, if $b\in {\mathbb R}_{+}^{N}$, one has $\tilde{h}\left( b\right) =\max_{q\in
		C}q^{\top }b=h\left( b\right) $. Now if $b_{z}<0$ for some $z$, one has
	clearly $\tilde{h}\left( b\right) =+\infty $. Hence 
	\begin{equation*}
		\tilde{h}\left( b\right) =\hat{h}\left( b\right) =h\left( b\right) +\iota
		_{\mathbb R_{+}^{N}}\left( b\right),
	\end{equation*}%
	where $\iota _{K}\left( b\right) =0$ if $b\in K$ and $\iota _{K}\left( b\right) =+\infty $ otherwise.
	This implies that $\hat{C}$ and $\tilde{C}$ have the same support function,
	and thus coincide.
\end{proof}

From Lemma~\ref{tympani}, it follows that:
\begin{lemma}\label{lem:convex-envelope-dominating}
\label{lem-4}The inequality $b^{\top }x\leq h_{X}\left( b\right) $ holds for all $%
b\in \mathbb{R}^N_{+}$ if and only if there is $\tilde{x}\in cch\left( X\right) $ with $%
x\leq \tilde{x}$.
\end{lemma}

\begin{proof}[Proof of Lemma~\ref{lem:convex-envelope-dominating}]
First, we assume $b^{\top }x\leq h_{X}\left( b\right) $ for all $b\geq 0$, and show that $x\in Y$
where $Y=\left\{ x^{\prime }:\exists \tilde{x}\in cch\left( X\right)
:x^{\prime }\leq \tilde{x}\right\} $. One has $X\subseteq Y$. Consider $%
h_{Y}\left( b\right) $ for $b\in R^{Z}$. First, if $b_{z}<0$
for some $z$, then $h_{Y}\left( b\right) =+\infty $. Next, if $%
b\geq 0$, then $h_{Y}\left( b\right) =h_{X}\left(
b\right) $. Indeed, one has $h_{X}\left( b\right) \leq \iota
_{Y}^{\ast }\left( b\right) $, but taking $y\in Y$ such that $b^{\top }y$
attains $h_{Y}\left( b\right) $, we have $\iota _{Y}^{\ast
}\left( b\right) =b^{\top }y$ and by definition of $y\in Y$, there is  $%
 \tilde{x}\in cch\left( X\right)$ such that $y\leq \tilde{x}$. Hence as $b\geq 0$%
, $b^{\top }y\leq b^{\top }\tilde{x}$, and as $\tilde{x}\in cch\left(
X\right) $, we get $b^{\top }\tilde{x}\leq h_{X}\left( b\right) 
$, thus $h_{Y}\left( b\right) \leq h_{X}\left(
b\right) $. As a result $h_{Y}\left( b\right) =h_{X}\left( b\right) $ as soon as $b\geq 0$, and we have that%
\[
b^{\top }x\leq h_{Y}\left( b\right) \text{ for all }b\in R^{Z}
\]%
and therefore, given that $Y$ is a closed convex set, this implies that $%
x\in Y$.

Conversely, assume there is $\tilde{x}\in cch\left( X\right) $ with $x\leq \tilde{x}$.  Then $\tilde{x}=\int_{0}^{1}x_{t}d\mu \left( t\right) $ where $\mu $ is a
probability measure on $\left[ 0,1\right] $ and $x_{t}\in X$. Then we have $%
x_{t}^{\top }b\leq h_{X}\left( b\right) $ and  thus $x^{\top }b\leq 
\tilde{x}^{\top }b=\int_{0}^{1}x_{t}^{\top }bd\mu \left( t\right) \leq
\int_{0}^{1}h_{X}\left( b\right) d\mu \left( t\right) = h_{X}\left( b\right) $.
\end{proof}

By combining Lemma~\ref{lem:char-subdiff} and Lemma~\ref{tympani}, we get:
\begin{lemma}\label{lem:diffcstartilde}
For a convex function $f: \mathbb R^N \to \mathbb R$, one has
	\begin{equation}\label{stony}
		\left\{ \tilde{q}\in
		\mathbb R^{N}:\exists q\in \partial f\left( p\right) ~s.t.~\tilde{q}\leq
		q\right\} = \left\{ q\in \mathbb R^{N}:q^{\top
		}b\leq \frac{d}{dt}f\left( p+tb\right) |_{0^{+}},\forall b\in
		\mathbb R_{+}^{N}\right\},
	\end{equation}%
and both these values coincide with $\widetilde{\partial f(p)}$.
\end{lemma}

	\begin{proof}[Proof of Lemma~\ref{lem:diffcstartilde}]
Let the function $h$ in the specification of $\hat C$ be given by $h\left( b\right) =\frac{d}{dt}	f\left( p+tb\right) |_{0^{+}}$.  From Claim 1, we then have $h(b) = \max _{q\in\partial f(p)}q^{\top}b$ and hence can take $C=\partial f(p)$ in the specification of $\tilde C$.  The equality of $\hat C$ and $\tilde C$ established in  Lemma \ref{tympani} then gives the required identity \eqref{stony}.
	\end{proof}
	
	\bigskip
\noindent In the sequel, we shall consider a pair of prices $p$ and $p^\prime$ in $\mathbb R^N$, and for a vector $b \in \mathbb R^N$, we define two vectors $b^{>}$ and $b^{\leq}$ in $\mathbb R^N$ such that
\begin{equation}\label{eq:border}
b_{z}^{\leq }=b_{z}\mathbf 1_{\left\{ p_{z}\leq p_{z}^{\prime }\right\}}\text{ and } b_{z}^{>}=b_{z}\mathbf 1_{\left\{ p_{z}>p_{z}^{\prime }\right\}}.
\end{equation}
	
\begin{lemma}\label{huddle}
A function $f:\mathbb R^N\rightarrow \mathbb R$ is submodular if and only if for any $p,p^\prime$ in $\mathbb R^N$,  $b\in R_{+}^{N}$ such that  $b^{>} \leq (p-p^\prime)^+$, one has
	\begin{equation*}
		f\left( p+b^{\leq }\right) +f\left( p^{\prime }+b^{>}\right) +f\left(
		p\wedge p^{\prime }\right) \leq f\left( p\right) +f\left( p^{\prime }\right)
		+f\left( p\wedge p^{\prime }+b\right) .
	\end{equation*}
\end{lemma}

\begin{proof}[Proof of Lemma~\ref{huddle}]
	Suppose $f$ is submodular. Then we have, by the submodularity of $f$, 
	\begin{equation*}
		f\left( p\right) +f\left( p\wedge p^{\prime }+b\right) \geq f\left( p\wedge
		\left( p\wedge p^{\prime }+b\right) \right) +f\left( p\vee \left( p\wedge
		p^{\prime }+b\right) \right).
	\end{equation*}%
However,  $p\wedge \left( p\wedge p^{\prime }+b\right) =p\wedge p^{\prime }+b^{>}$
	and $p\vee \left( p\wedge p^{\prime }+b\right) =p+b^{\leq }$, and so%
	\begin{equation*}
		f\left( p\right) +f\left( p\wedge p^{\prime }+b\right) \geq f\left( p\wedge
		p^{\prime }+b^{>}\right) +f\left( p+b^{\leq }\right) .
	\end{equation*}
	This implies%
	\begin{equation*}
		f\left( p\right) +f\left( p^{\prime }\right) +f\left( p\wedge p^{\prime
		}+b\right) \geq f\left( p\wedge p^{\prime }+b^{>}\right) +f\left( p^{\prime
		}\right) +f\left( p+b^{\leq }\right).
	\end{equation*}%
	Again by the submodularity of $f$, we have 
	\begin{equation*}
		f\left( p\wedge p^{\prime }+b^{>}\right) +f\left( p^{\prime }\right) \geq
		f\left( \left( p\wedge p^{\prime }+b^{>}\right) \wedge p^{\prime }\right)
		+f\left( \left( p\wedge p^{\prime }+b^{>}\right) \vee p^{\prime }\right) ,
	\end{equation*}%
	and hence 
	\begin{eqnarray}
		&&f\left( p\right) +f\left( p^{\prime }\right) +f\left( p\wedge p^{\prime
		}+b\right)  \label{almost-final} \\
		&\geq &f\left( \left( p\wedge p^{\prime }+b^{>}\right) \wedge p^{\prime
		}\right) +f\left( \left( p\wedge p^{\prime }+b^{>}\right) \vee p^{\prime
		}\right) +f\left( p+b^{\leq }\right).  \notag
	\end{eqnarray}%
	But $\left( p\wedge p^{\prime }+b^{>}\right) \wedge p^{\prime }=p\wedge
	p^{\prime }$ and $\left( p\wedge p^{\prime }+b^{>}\right) \vee p^{\prime
	}=p^{\prime }+b^{>}$, and therefore~(\ref{almost-final}) becomes%
	\begin{equation*}
		f\left( p\right) +f\left( p^{\prime }\right) +f\left( p\wedge p^{\prime
		}+b\right) \geq f\left( p\wedge p^{\prime }\right) +f\left( p^{\prime
		}+b^{>}\right) +f\left( p+b^{\leq }\right),
	\end{equation*}%
	giving the required result.  
	
	Conversely, assume we have for all $p,p^{\prime }$ and $b\geq 0$ that%
	\begin{equation*}
		f\left( p\right) +f\left( p^{\prime }\right) +f\left( p\wedge p^{\prime
		}+b\right) \geq f\left( p+b^{\leq }\right) +f\left( p^{\prime}+b^>
		\right) +f\left( p\wedge p^{\prime }\right).
	\end{equation*}%
	Choose $b$ to be specified by $b^{\leq }=\left( p^{\prime }-p\right) ^{+}$ and $b^{>}=0$. We have%
	\begin{eqnarray*}
		p\wedge p^{\prime }+b&=&p^{\prime } \\ 
			p+b^{\leq }&=&p\vee p^{\prime }%
			\end{eqnarray*}%
	and thus%
	\begin{equation*}
		f\left( p\right) +f\left( p^{\prime }\right) \geq f\left( p\vee p^{\prime
		}\right) +f\left( p\wedge p^{\prime }\right) ,
	\end{equation*}
	giving the submodularity of $f$.
\end{proof}

\begin{lemma}\label{lemma:diff-char-submod}
A convex function $f:\mathbb R^N\rightarrow \mathbb R$ is  submodular if
and only if for any $b\in \mathbb R_{+}^{N}$ 
\begin{equation}\label{shuttle}
	\frac{d}{dt}f\left( p+tb^{\leq }\right) |_{0^{+}}+\frac{d}{dt}%
	f\left( p^{\prime }+tb^{>}\right) |_{0^{+}}\leq \frac{d}{dt}f\left( p\wedge p^{\prime }+tb\right) |_{0^{+}},
\end{equation}
where $b^{\leq}$ and $b^{>}$ are defined in equation~\eqref{eq:border}.
\end{lemma}

\begin{proof}[Proof of Lemma~\ref{lemma:diff-char-submod}]
Applying Lemma \ref{huddle}, if $f$ is submodular it follows that for $t\ge 0$, we have%
\begin{equation*}
	f\left( p+tb^{\leq }\right) -f\left( p\right) +f\left( p^{\prime }+tb^{>}\right) -f\left( p^{\prime }\right) \leq
	f\left( p\wedge p^{\prime }+tb\right) -f\left( p\wedge
	p^{\prime }\right),
\end{equation*}
and thus, using the convexity of $f$, it follows that%
\begin{equation*}
	\frac{d}{dt}f\left( p+tb^{\leq }\right) |_{0^{+}}+\frac{d}{dt}%
	f\left( p^{\prime }+tb^{>}\right) |_{0^{+}}\leq \frac{d}{dt}f\left( p\wedge p^{\prime }+tb\right) |_{0^{+}}.
\end{equation*}
The converse holds by integration over $t\in \left[ 0,1\right] $.
\end{proof}

\paragraph{Theorem~\ref{spice}, direct implication.} If $f$ is submodular then $\partial f\left(
p\right)$ exhibits unified gross substitutes.  

\begin{proof}[Proof of the direct implication of Theorem~\ref{spice}]
	Assume $f$ is
	submodular.
	Take $q\in \partial f\left( p\right) $ and $q^{\prime }\in \partial
	f\left( p^{\prime }\right) $. We want to show that there exists $%
	q^{\wedge }\in \partial f\left( p\wedge p^{\prime }\right) $ such
	that 
	\begin{eqnarray*}
		p_{z} \leq p_{z}^{\prime }&\implies& q_{z}\leq q_{z}^{\wedge } \\
		p_{z} >p_{z}^{\prime }&\implies& q_{z}^{\prime }\leq q_{z}^{\wedge }.
	\end{eqnarray*}
	To show this, we need to show that $q\mathbf 1_{\mathcal{Z}^{\leq }}+q^{\prime }\mathbf 1_{\mathcal{Z}^{>}}\le q^{\wedge}$, i. e, we need to show that $q\mathbf 1_{\mathcal{Z}^{\leq }}+q^{\prime }\mathbf 1_{\mathcal{Z}^{>}}\in 
	\widetilde{\partial f}\left( p\wedge p^{\prime }\right) $, where the tilde notation $\widetilde{\partial f }$ was introduced in~\eqref{eq:Ctilde}.  By Lemma~\ref{lem:diffcstartilde}, it suffices to show that%
	\begin{equation*}
		\forall b\in \mathbb R_{+}^{N}:\left( q\mathbf 1_{\mathcal{Z}^{\leq }}+q^{\prime }\mathbf1_{\mathcal{Z}^{>}}\right)
		^{\top }b\leq \frac{d}{dt}f\left( p\wedge p^{\prime }+tb\right)
		|_{0^{+}}.
	\end{equation*}
	In order to do this, take $b\in \mathbb R_{+}^{N}$ and express that $q \in \partial f(p)$ and $q^\prime \in \partial f(p^\prime)$ by writing 
	\begin{equation*}
		q^{\top }b^{\leq } \leq \frac{d}{dt}f\left( p+tb^{\leq }\right)
		|_{0^{+}}\text{ and } 
		q^{\prime \top }b^{>} \leq \frac{d}{dt}f\left( p^{\prime
		}+tb^{>}\right) |_{0^{+}},
	\end{equation*}%
	and then note that, by summation of these two inequalities we get
	\begin{equation*}
		\left( q\mathbf1_{\mathcal{Z}^{\leq }}+q^{\prime }\mathbf1_{\mathcal{Z}^{>}}\right) ^{\top }b\leq \frac{d}{dt}%
		f\left( p+tb^{\leq }\right) |_{0^{+}}+\frac{d}{dt}f\left(
		p^{\prime }+tb^{>}\right) |_{0^{+}},
	\end{equation*}%
	and so by Lemma~\ref{lemma:diff-char-submod}, we have 
	\begin{equation*}
		\forall b\in \mathbb R_{+}^{N}:\left( q\mathbf1_{\mathcal{Z}^{\leq }}+q^{\prime }\mathbf1_{\mathcal{Z}^{>}}\right)
		^{\top }b\leq \frac{d}{dt}f\left( p\wedge p^{\prime }+tb\right)
		|_{0^{+}},
	\end{equation*}%
	and hence (again, by Lemma~\ref{lem:diffcstartilde}) $q\mathbf1_{\mathcal{Z}^{\leq }}+q^{\prime }\mathbf1_{\mathcal{Z}^{>}}\in \widetilde{\partial f}%
	\left( p\wedge p^{\prime }\right)$ as required.
	
We set $\hat{f} = f(-p)$, and we introduce $\hat{p} = p^\prime$ and $\hat{p}^\prime = p$. We note that $\hat{f}$ is submodular, and we apply the previous claim to $\hat{p}$ and $\hat{p'}$ to get the existence of $q^\vee \in \partial f(p\vee p^\prime)$ such that 
\begin{eqnarray*}
	p_{z} \leq p_{z}^{\prime }&\implies& q_{z}\geq q_{z}^{\vee } \\
	p_{z} >p_{z}^{\prime }&\implies& q_{z}^{\prime }\geq q_{z}^{\vee }.
\end{eqnarray*}

\end{proof}

\paragraph{Theorem~\ref{spice}, backward implication.} The converse holds, i.e. if $\partial f\left(
p\right) $ satisfies unified gross substitutes, then $f\left( p\right) $ is submodular.

\begin{proof}[Proof of the backward implication of Theorem~\ref{spice}]
	Assume unified gross substitutes holds. Then for $q\in \partial f\left( p\right) $ and $%
	q^{\prime }\in \partial f\left( p^{\prime }\right) $, there exists $%
	q^{\wedge }\in \partial f\left( p\wedge p^{\prime }\right) $ such
	that%
	\begin{equation*}
		q\mathbf1_{\mathcal{Z}^{\leq }}+q^{\prime }\mathbf1_{\mathcal{Z}^{>}}\leq q^{\wedge }.
	\end{equation*}%
	Hence $q\mathbf 1_{\mathcal{Z}^{\leq }}+q^{\prime }\mathbf 1_{\mathcal{Z}^{>}}\in \widetilde{\partial f}%
	\left( p\wedge p^{\prime }\right) $, and therefore (using Lemma~\ref{tympani}), for all $q\in \partial
	f\left( p\right) $ and $q^{\prime }\in \partial f\left(
	p^{\prime }\right) $ 
	\begin{equation*}
		\forall b\in {\mathbb R}_{+}^{Z}:q^{\top }b^{\leq }+q^{\prime \top }b^{<}\leq \frac{d}{%
			dt}f\left( p\wedge p^{\prime }+tb\right) |_{0^{+}},
	\end{equation*}%
	and thus, by maximizing over $q\in \partial f\left( p\right) $ and $%
	q^{\prime }\in \partial f\left( p^{\prime }\right) $, we get 
	\begin{equation*}
		\frac{d}{dt}f\left( p+tb^{\leq }\right) |_{0^{+}}+\frac{d}{dt}%
		f\left( p^{\prime }+tb^{>}\right) |_{0^{+}}\leq \frac{d}{dt}c^{\ast
		}\left( p\wedge p^{\prime }+tb\right) |_{0^{+}}
	\end{equation*}%
	for all $b\geq 0$, hence (by Lemma~\ref{lemma:diff-char-submod}) $f$ is submodular.
\end{proof}

\subsection{Proof of Theorem \ref{camel}}\label{iron}

Let $q\in \mathtt Q\left(p\right)$ and $q^{\prime}\in \mathtt Q\left(p^{\prime}\right)$.  Rewriting \eqref{lester}--\eqref{flatt}, it suffices for unified gross substitutes to show that there exists $q^{\vee}\in \mathtt Q\left(p\vee p^{\prime}\right)$ and $q^{\wedge}\in \mathtt Q\left(p\wedge p^{\prime}\right)$ such that for all $z\in\mathcal{Z}$,
\begin{eqnarray*}
	\mathbf 1_{\left\{ z\in\mathcal{Z}^{\leq}\right\}} q_{z}+\mathbf1_{\left\{ z\in\mathcal{Z}%
		^{>}\right\} }q_{z}^{\prime}&\leq& q_{z}^{\wedge}\\
	\mathbf1_{\left\{ z\in\mathcal{Z}^{\leq}\right\} }q_{z}^{\prime}+\mathbf1_{\left\{ z\in%
		\mathcal{Z}^{>}\right\} }q_{z}&\ge& q_{z}^{\vee},
\end{eqnarray*}
where we have defined $\mathcal{Z}^{\leq}=\left\{ z\in\mathcal{Z}:p_{z}\leq p_{z}^{\prime}\right\} $ and $\mathcal{Z}^{>}=\left\{ z\in\mathcal{Z}:p_{z}>p_{z}^{\prime}\right\} $.

First, note that $\mu_{xz}>0$ and $p_{z}\leq p_{z}^{\prime}$ implies $p^{\prime}_x \geq G_{xz} (p^{\prime}_z) \geq G_{xz} (p_z) = p_x$, which implies
\begin{equation}\label{cymbal}
	\mu_{xz}\mathbf1_{\left\{ z\in\mathcal{Z}^{\leq}\right\}}
	\leq\mu_{xz}\mathbf1_{\left\{ x\in\mathcal{Z}^{\leq}\right\}}.
\end{equation}
Similarly,  $\mu_{xz}^{\prime}>0$ and $p_{z}>p_{z}^{\prime}$ implies
$p_x \geq G_{xz} (p_z) > G_{xz} (p^{\prime}_z) = p^{\prime}_x $, thus 
\begin{equation}\label{triangle}
	\mu_{xz}^{\prime}\mathbf1_{\left\{ z\in\mathcal{Z}^{>}\right\}}
	\leq\mu_{xz}^{\prime}\mathbf1_{\left\{ x\in\mathcal{Z}^{>}\right\}}.
\end{equation}
Given these results, set:
\begin{eqnarray*}
	\mu_{xz}^{\wedge}&=&\mathbf1_{\left\{ x\in\mathcal{Z}^{\leq}\right\}}
	\mu_{xz}+\mathbf1_{\left\{ x\in\mathcal{Z}^{>}\right\} }\mu_{xz}^{\prime}\\
	q_{z}^{\wedge}&=&\sum_{x}\mu_{xz}^{\wedge}-\sum_{y}\mu_{zy}^{\wedge}.
\end{eqnarray*}
We have $\mu_{xy}^{\wedge}>0$ implies $p_x=G_{xy}(p_y)$,
and, using \eqref{cymbal} and \eqref{triangle} for the inequality, we obtain
\begin{align*}
	q_{z}^{\wedge} & =\sum_{x}(\mathbf1_{\left\{ x\in\mathcal{Z}^{\leq}\right\}} \mu_{xz}+\mathbf1_{\left\{
		x\in\mathcal{Z}^{>}\right\}} \mu_{xz}^{\prime})  -\sum_{y}(\mathbf1_{\left\{ z\in\mathcal{Z}^{\leq}\right\}} \mu_{zy}+\mathbf1_{\left\{
		x\in\mathcal{Z}^{>}\right\}} \mu_{zy}^{\prime}) \\
	& \geq\sum_{x}(\mathbf1_{\left\{ z\in\mathcal{Z}^{\leq}\right\}} \mu_{xz}+\mathbf1_{\left\{
		z\in\mathcal{Z}^{>}\right\}} \mu_{xz}^{\prime})  -\sum_{y}(\mathbf1_{\left\{ z\in\mathcal{Z}^{\leq}\right\}} \mu_{zy}+\mathbf1_{\left\{
		z\in\mathcal{Z}^{>}\right\}} \mu_{zy}^{\prime}) \\
	& =\mathbf1_{\left\{ z\in\mathcal{Z}^{\leq}\right\}} q_{z}+\mathbf1_{\left\{ z\in\mathcal{Z}%
		^{>}\right\}} q_{z}^{\prime},
\end{align*}
giving the required result for $q^{\wedge}$.
A similar argument shows that $\mathbf1_{\left\{ z\in\mathcal{Z}^{\leq}\right\}}
q_{z}^{\prime}+\mathbf1_{\left\{ z\in\mathcal{Z}^{>}\right\}} q_{z}^{\prime}\geq
q_{z}^{\vee}$.

\subsection{Matching and Equilibrium Flows}\label{chatter}
We associate an equilibrium flow problem with this matching market, generalized to accommodate unmatched agents.  See  figure~\ref{fig:matching-with-singles}.

\begin{figure}[t]
    \centering
    \includegraphics[width=0.35\textwidth]{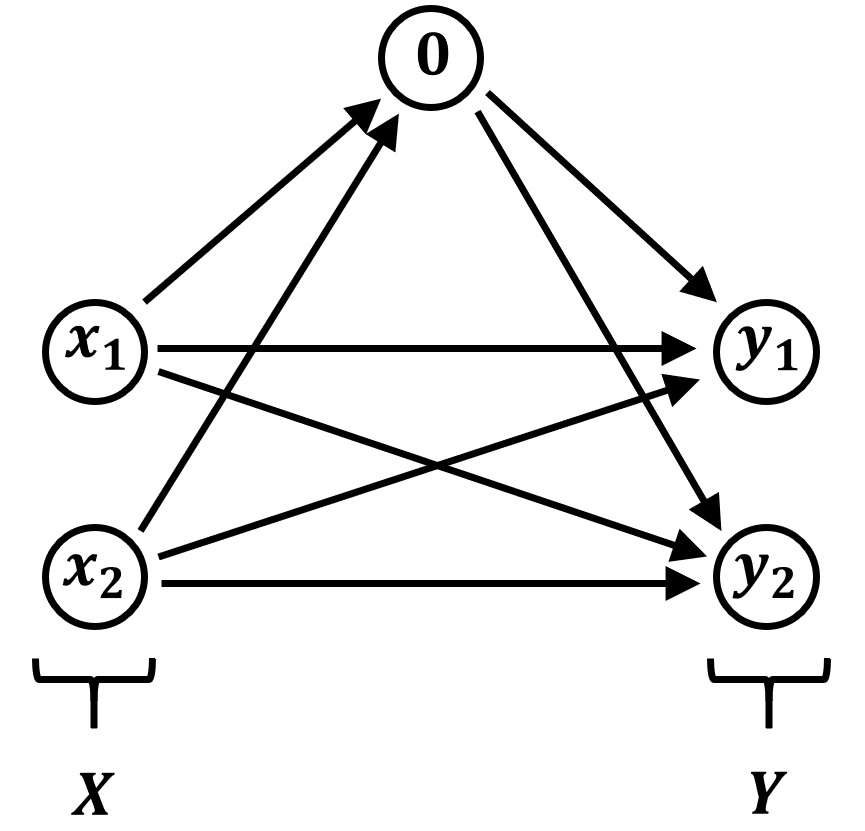}
    \caption{Network associated with a bipartite matching problem with singles.}
    \label{fig:matching-with-singles}
\end{figure}

\begin{eqnarray*}
	\mathcal{Z}&=&\mathcal{X}\cup\mathcal{Y}\cup\{0\}\\
	\mathcal{A}&=&((\mathcal{X}\cup \{0\})\times(\mathcal{Y}\cup\{0\}))\setminus \{(0,0)\}\\
	q_z &=& - n_z,~~~~~z\in \mathcal X\\
	q_z &=& m_z,~~~~~z\in \mathcal Y.
\end{eqnarray*}
Given a price $p$, we define the connection function
\begin{eqnarray*}
	G_{xy}\left(p_y\right)&=& \mathcal{U}_{xy}\circ \mathcal{V}_{xy}^{-1}\left(-p_{y}\right) \text{ for }x\in\mathcal{X},y\in\mathcal{Y} \\
	G_{x0}\left(p\right)&=&p_{0}\text{ for }x\in\mathcal{X} \\
	G_{0y}\left(p\right)&=&p_{y}\text{ for }y\in\mathcal{Y}.
\end{eqnarray*}

\noindent As required, $G_{xy}(p_y)$ is increasing in $p_y$.  We adopt the normalization $p_{0}=0$. We now establish a relationship between stable matchings of the matching market and equilibria of the equilibrium flow problem.   First, let $(q,\mu,p)$ be an equilibrium of the equilibrium flow problem.  
To obtain a stable matching $(\mu,w)$, choose $w_{xy}$ so that $\mathcal{V}%
_{xy}^{-1}\left(-p_{y}\right)\leq w_{xy}\leq\mathcal{U}_{xy}^{-1}\left(p_{x}%
\right)$.  Because we have $p_x \geq G_{xy} (p_y)$ in equilibrium, this is possible.  Then for worker $x$ we have
\begin{eqnarray*}
	\mu_{xy}>0 &\implies& p_x = G_{xy}\left(p_y\right) \\
	&\implies&w_{xy} = \mathcal U^{-1}_{xy}(p_x) \\
	p_x \geq G_{xz}(p_z) \le 0&\implies&w_{xz} \le \mathcal U^{-1}_{xz}(p_x)\\
	&\implies &\mathcal U_{xy}(w_{xy})\ge \mathcal U_{xz}(w_{xz}),
\end{eqnarray*}
which gives
\[
y\in\arg\max_{y\in\mathcal{Y}}\mathcal{U}%
_{xy}\left(w_{xy}\right),
\]
as needed.  The argument for firms is similar, giving a stable matching.	

Conversely, suppose we have a stable matching $(\mu,w)$.  Define $q_z = - n_x$ for $z\in \mathcal X$ and $q_z = m_y$ for $z\in \mathcal Y$.  We identify prices $p$ such that $(q,\mu,p)$ is an equilibrium of the equilibrium flow problem.  Define the indirect utilities
\begin{eqnarray*}
	u_{x}&=&\max_{x\in\mathcal{X}}\left\{\mathcal{U}_{xy}\left(w_{xy}\right),0\right\} \\
	v_{y} &=&\max_{y\in%
		\mathcal{Y}}\left\{ \mathcal{V}_{xy}\left(w_{xy}\right),0\right\}.
\end{eqnarray*}
Then define $p$ by $p_{x}=u_{x}$ if $x\in\mathcal{X}$, $p_{y}=-v_{y}$ if $y\in\mathcal{Y}$ and $p_{0}=0$, and define $q_z = - n_x$ for $z\in \mathcal X$ and $q_z = m_y$ for $z\in\mathcal Y$.  Define the connection functions $G_{xy}$ as above.  If $\mu_{xy} = 0$, then it follows from the stability condition for a stable matching that $p_x = G_{xy}(p_y)$.  For other pairs $(x,y)$, the stability condition implies that there is no wage $w_{xy}$ at which $x$ and $y$ can match and obtain utilities in excess of $u_x$ and $v_y$, which is equivalent to the statement that $p_x \geq G_{xy}(p_y)$.  We thus have an equilibrium flow.

\subsection{Proof of Theorem \ref{th_ntu_matching}}\label{proof_th_ntu_matching}

Statement (i): Assume $\mu $ is a stable matching. Let $y\in \mathcal{Y}$ and let $%
x\in \mathcal{X}_{0}$ be the match of $y$ under $\mu $. We need to consider two
cases:

\bigskip

\noindent First case: Let $x\in \mathcal{X}$.   We have that $v_{y}^{\mu }=\gamma _{xy}$,
and there cannot be another $y^{\prime }\in \mathcal{Y}$ with $\gamma
_{xy^{\prime }}\geq v_{y^{\prime }}^{\mu }$ and $\alpha _{xy^{\prime
}}>u_{x}^{\mu }=\alpha _{xy}$; for if this were the case, $y\neq y^{\prime }$%
, and hence $x$ is not matched with $y^{\prime }$, and thus $\gamma
_{xy^{\prime }}\geq v_{y^{\prime }}^{\mu }$ would imply $\gamma _{xy^{\prime
}}>v_{y^{\prime }}^{\mu }$. As a result, 
\begin{equation}
y\in \arg \max_{y\in \mathcal{Y}}\left\{ \alpha _{xy}:\gamma _{xy}\geq
v_{y}^{\mu },\alpha _{x0}\right\} .  \label{res1}
\end{equation}

\noindent Now assume that $y\in \arg \max_{y}\left\{ \alpha _{x^{\prime }y}:\gamma
_{x^{\prime }y}\geq v_{y},\alpha _{x^{\prime }0}\right\} $ for some $%
x^{\prime }\neq x$; but then $x^{\prime }y$ would be a blocking pair;
further, $v_{y}^{\mu }=\gamma _{xy}>\gamma _{0y}$. As a result, for any $%
x^{\prime }\in \mathcal{X}\backslash \left\{ x\right\} $, we have%
\begin{equation}
y\notin \arg \max_{y}\left\{ \alpha _{x^{\prime }y}:\gamma _{x^{\prime
}y}\geq v_{y},\alpha _{x^{\prime }0}\right\}.  \label{res2}
\end{equation}
Finally, because $\mu $ is stable and $x\neq 0$, we have 
\begin{equation}
v_{y}^{\mu }=\gamma _{xy}>\gamma _{0y}.  \label{res3}
\end{equation}
By putting (\ref{res1}), (\ref{res2}), and (\ref{res3}) together, we get
that $Q_{y}\left( v^{\mu }\right) =0$.

\bigskip

\noindent Second case: Let $x=0$.  Then $v_{y}^{\mu }=\gamma _{0y}$ and a similar logic
as above shows that for any $x^{\prime }\in \mathcal{X}$, we have%
\begin{equation*}
y\notin \arg \max_{y}\left\{ \alpha _{x^{\prime }y}:\gamma _{x^{\prime
}y}\geq v_{y},\alpha _{x^{\prime }0}\right\},  
\end{equation*}%
and as a result we get that $Q_{y}\left( v^{\mu }\right) =0$.

\bigskip

\noindent Statement (ii): Conversely, assume $Q\left( v\right) =0$. Then for $x\in \mathcal{X}$
and $y\in \mathcal{Y}$, we define 
\begin{equation*}
\left\{ 
\begin{array}{l}
\mu _{xy}=\mathbf 1_{\left\{ y\in \arg \max_{y}\left\{ \alpha _{xy}:\gamma _{xy}\geq
v_{y},\alpha _{x0}\right\} \right\} } \text{for }x\in \mathcal{X}\text{ and }%
y\in \mathcal{Y}_{0}\text{,} \\ 
\mu _{0y}=\mathbf 1_{\left\{ \gamma _{0y}\geq v_{y}\right\}} \text{ for }y\in \mathcal{Y%
}\text{.}%
\end{array}%
\right. 
\end{equation*}%
We have $\sum_{x\in \mathcal{X}_{0}}\mu _{xy}=1$ by assumption and it is
straightforward to see that $\sum_{y\in \mathcal{Y}_{0}}\mu _{xy}=1$ by the strict
preferences assumption.

We need to show that $\mu $ is a stable matching. First, let us show that
there is no blocking pair. By contradiction, assume $xy$ is a blocking pair.
Consider $x^{\ast }\in \mathcal{X}_{0}$, the match of $y$ under $\mu $, and $%
y^{\ast }\in \mathcal{Y}_{0}$, the match of $x$ under $\mu $. Because $xy$ is
a blocking pair, we have $\gamma _{xy}>\gamma _{x^{\ast }y}$ and $\alpha
_{xy}>\alpha _{xy^{\ast }}$.

We show that $\gamma _{xy}>v_{y}$. Indeed, if $x^{\ast }\neq 0$ we have that because $x^{\ast }$ is matched with $y$, then $%
y\in \arg \max_{y}\left\{ \alpha _{x^{\ast }y}:\gamma _{x^{\ast }y}\geq
v_{y},\alpha _{x^{\ast }0}\right\} $, and $\gamma _{x^{\ast }y}\geq v_{y}$, and thus $\gamma _{xy}>\gamma _{x^{\ast
}y}\geq v_{y}$. Otherwise, if $x^{\ast }=0$, we have $\mu _{0y}=1$, and hence $%
\gamma _{0y}\geq v_{y}$; but we have $\gamma _{xy}>\gamma _{0y}$, because $xy$
is a blocking pair, and thus $\gamma _{xy}>\gamma _{0y}\geq v_{y}$. In either
case, $\gamma _{xy}>v_{y}$ as announced.

As $\alpha _{xy^{\ast }}=\max_{y}\left\{ \alpha _{xy}:\gamma _{xy}\geq
v_{y},\alpha _{x0}\right\} $ and $\gamma _{xy}>v_{y}$, it follows that $%
\alpha _{xy^{\ast }}\geq \alpha _{xy}$, a contradiction.

Finally, we need to show that there is no blocking individual. Assume $xy$
is a matched pair. As $\mu _{xy}=1$, we have $\alpha _{xy}\geq \alpha _{x0}$
by definition of $\mu $. Next, we have $\mu _{0y}=0$ therefore $\gamma
_{0y}<v_{y}$ and in particular $\gamma _{0y}\leq v_{y}$.

\subsection{Proof of Theorem \ref{th_ntu_matching_m0}}\label{proof_th_ntu_matching_m0}

It is straightforward to verify that $\mathtt Q$ satisfies weak gross substitutes, and hence it
defines a point-valued correspondence for which unified gross substitutes
holds, by Property \ref{prop:wgs-implies-ugs}. Next, we show that is it a M0-correspondence by showing that it satisfies
monotone total output. We have 
\begin{eqnarray*}
\sum_{y\in Y}Q_{y}\left( v\right)  &=&\left\vert Y\right\vert -\sum_{y\in
Y}\sum_{x\in X}\mathbf 1_{\left\{ y\in \arg \max_{y}\left\{ \alpha _{xy}:\gamma
_{xy}\geq v_{y},\alpha _{x0}\right\} \right\}} -\sum_{y\in Y}\mathbf 1_{\left\{ \gamma
_{0y}\geq v_{y}\right\} } \\
&=&\left\vert Y\right\vert -\sum_{x\in X}\sum_{y\in Y}\mathbf 1_{\left\{ y\in \arg
\max_{y}\left\{ \alpha _{xy}:\gamma _{xy}\geq v_{y},\alpha _{x0}\right\}
\right\}} -\sum_{y\in Y}\mathbf 1_{\left\{ \gamma _{0y}\geq v_{y}\right\}}  \\
&=&\left\vert Y\right\vert -\left\vert X\right\vert -\sum_{x\in X}\mathbf 1_{\left\{
0\in \arg \max_{y}\left\{ \alpha _{xy}:\gamma _{xy}\geq v_{y},\alpha
_{x0}\right\} \right\}} -\sum_{y\in Y}\mathbf 1_{\left\{ \gamma _{0y}\geq v_{y}\right\}} 
\\
&=&\left\vert Y\right\vert -\left\vert X\right\vert -\sum_{x\in X}\mathbf 1_{\left\{
\alpha _{x0}\geq \max_{y\in Y}\left\{ \alpha _{xy}:\gamma _{xy}\geq
v_{y}\right\} \right\}} -\sum_{y\in Y}\mathbf 1_{\left\{ \gamma _{0y}\geq v_{y}\right\}}. 
\end{eqnarray*}%
We then note that $\mathbf 1_{\left\{ \alpha _{x0}\geq \max_{y\in Y}\left\{ \alpha _{xy}:\gamma
_{xy}\geq v_{y}\right\} \right\}} $ is nondecreasing in each $v_{y}$, and similarly  $\mathbf 1_{\left\{ \gamma _{0y}\geq v_{y}\right\}} $ is also nondecreasing
in each $v_{y}$. Hence monotone total output holds, and nonreversingness
follows.

\subsection{Proof of Theorem \ref{thm:ugs-for-hedonic}}\label{app:hedonic}

We show that the hedonic pricing problem is an equilibrium flow problem, as illustrated in Figure~\ref{fig:hedonic-model}.

\begin{figure}[ht]
    \centering
    \includegraphics[width=0.35\textwidth]{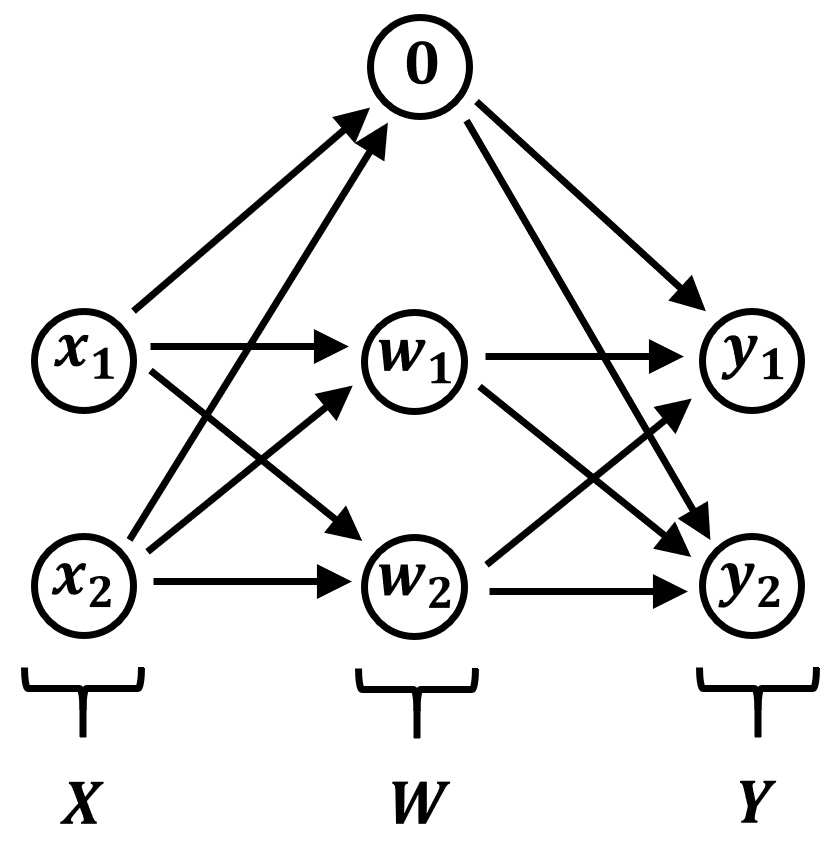}
    \caption{Network associated with a hedonic model.}
    \label{fig:hedonic-model}
\end{figure}

Consider $\mathcal{Z}=\mathcal{X}\cup \mathcal{Y}\cup \mathcal{W}$ and $%
\mathcal{Z}_{0}=\mathcal{Z}\cup \left\{ 0\right\} $, where $0$ is an
additional node. Denote as well $\mathcal{W}_{0}=\mathcal{W}\cup \left\{
0\right\} $. The set of arcs $\mathcal{A}$ is given by%
\begin{equation*}
	\mathcal{A}=(\mathcal X\times \mathcal{W}_{0})\cup (\mathcal{W}_{0}\times \mathcal{Y}).
\end{equation*}
To define the connection functions, first let $p_{x} = u_{x}$ for all $x\in \mathcal{X}$, $p_{w} = p_{w}$ for all $w\in \mathcal{W}$ and $p_{y} = -v_{y}$ for all $y\in \mathcal{Y}$.

Then the connections functions are given by $G_{xw}\left( p_w\right) =\pi_{xw}\left( p_{w}\right)$, $G_{x0}\left( p_0\right) =p_{0}$, $G_{wy}\left( p_y\right) = s^{-1}_{yw}\left( - p_{y}\right)$ and $G_{0y}\left( p_y\right) =  p_{y}$.

The equilibrium stocks are given by $q_{w} = \sum_{x\in \mathcal{X}}\mu _{x0} - \sum_{y\in \mathcal{Y}}\mu _{0y}$ for all $w\in \mathcal{W}$, $q_{x} = -n_{x}$ for all $x\in \mathcal{X}$, $q_{w} = 0$ for all $w\in \mathcal{W}$ and $q_{y} = m_{y}$ for all $y\in \mathcal{Y}$.

We can then reformulate the equilibrium conditions of Definition \ref{book} as an equilibrium flow on this network, along with the normalization $p_0=0$.

\bibliographystyle{plain}
\bibliography{larrybib}

\clearpage


\begin{center}
	{\large Monotone Comparative Statics 		for Equilibrium Problems
	
\medskip	
	
Online Appendix}
	
	\bigskip
	
Alfred Galichon$^{\dag }$, Larry Samuelson{\small $^{\flat }$}, and Lucas Vernet{\small $^{\S }$}
	
	\bigskip
	
	\today
	
	\bigskip

\end{center}

\section{Online Appendix}

\subsection{Weak or Strong Inequalities?}\label{par:weak-strong}

One might be tempted to strengthen Definition \ref{duck}  by asking for weak inequalities in the antecedents of both \eqref{lester} and \eqref{flatt}.  To see why we retain the asymmetry, consider the correspondence $\mathtt Q$ defined by
\begin{equation}\label{vesper}
	\mathtt Q\left(p\right)=\arg\max_{q\in\mathbb{R}_{+}^{N}}\left\{ \sum_{z=1}^Np_{z}q_{z}:\sum_{z=1}^Nq_{z}=1\right\} .
\end{equation}
We can think of an agent who must select an outcome from the set $\{1,\ldots, N\}$, receiving a  payoff $p_z$ when selecting outcome $z$.  
The agent then maximizes her expected payoff by predicting the  outcome with the highest payoff, if unique, and otherwise choosing any mixture of the set of payoff-maximizing outcomes.  The correspondence $\mathtt Q(p)$ describes this payoff-maximizing strategy.%
\footnote{Alternatively, suppose we interpret $p\in P\subseteq \mathbb R^N_+$ as a vector of realizations of random utilities, and suppose that demand is given by a multinomial logit function.  Then $q\in \mathtt Q(p)$ if 
	\[
	q_z = \frac{{\rm exp}\frac{p_z}{T}}{\sum_{\tilde z=1}^N{\rm exp}\frac{p_{\tilde z}}{T}},~~~~z=1,\ldots,N.
	\]
	for some ``rationality''  parameter $T>0$ that is fixed as part of the logit specification.  In the limit as $T\rightarrow 0$, we obtain a demand correspondence,  in which a vector $q$ is optimal if and only if it exhausts  the budget constraint and $q_z>0\Longleftrightarrow p_z = \max_{\tilde z\in \{1,\ldots,N\}}p_{\tilde z}$, as specified by \eqref{vesper}.}

One should expect this problem to give rise to substitutes.  If the current choice is outcome $z$, and then the payoff attached to outcome $z'$ increases, the choice should either remain $z$ (if the payoff on $z'$ is still too small) or switch to $z'$, making outcomes $z$ and $z'$ substitutes.  

We confirm at the end of this section that \eqref{vesper} satisfies unified gross substitutes.
To see that it fails the stronger condition, consider  the prices $p$ and $p'$ and allocations $q\in \mathtt Q(p)$ and $q'\in \mathtt Q(p')$ given by
\begin{equation*}
	\begin{array}{ccccccc}
		z && p & q& p^{\prime } & q^{\prime } & p\vee p^{\prime } \\ \hline
		1: && 1 & 1 & 1 & 0 & 1 \\
		2: && 1 & 0 & 1 & 1 & 1%
	\end{array}%
	~~~.
\end{equation*}
We have the following implications of \eqref{lester} and the weak-inequality version of \eqref{flatt}:
\begin{eqnarray*}
	1=p_{1}\leq p_{1}^{\prime }=1 &\implies &q_{1}^{\vee }\leq q_{1}^{\prime }=0\\
	1=p_{2}^{\prime }\leq p_{2}=1 &\implies &q_{2}^{\vee }\leq q_{2}=0.
\end{eqnarray*}%
But then we must have $q_{1}^{\vee }=q_{2}^{\vee }=0$, which contradicts the requirement that $q^{\vee}\in \mathtt Q(p\vee p')$.
The correspondence $\mathtt Q$ thus fails the proposed stronger formulation of unified gross substitutes.

To show that \eqref{vesper} satisfies unified gross substitutes, consider $q\in \mathtt Q\left( p\right) $ and $q^{\prime }\in \mathtt Q\left( p^{\prime}\right) $. We would like to show that:

(i) there is $q^{\vee }\in \mathtt Q\left( p\vee p^{\prime }\right) $ such that
\begin{equation}
	q_{z}^{\vee }\leq q_{z}^{\prime }\bm1_{\left\{ p_{z}\leq p_{z}^{\prime }\right\}}
	+q_{z}\bm1_{\left\{ p_{z}>p_{z}^{\prime }\right\}}, ~{\rm and}   \label{req2}
\end{equation}

(ii) there is $q^{\wedge }\in \mathtt Q\left( p\wedge p^{\prime }\right) $ such that
\begin{equation}
	q_{z}\bm 1_{\left\{ p_{z}\leq p_{z}^{\prime }\right\}} +q_{z}^{\prime }\bm1_{\left\{
		p_{z}>p_{z}^{\prime }\right\}} \leq q_{z}^{\wedge }.  \label{req1}
\end{equation}

\noindent [{\sc Step} (i)]  To establish (i), we consider two cases:

\begin{itemize}
	
	\item[(a)] $\max p_{z}>\max p_{z}^{\prime }$. Then set $q^{\vee }=q$, and notice that $%
	q_{z}^{\vee }=q_{z}>0$ implies $p_{z}>p_{z}^{\prime }$, giving $q_{z}^{\prime
	}\bm1_{\left\{ p_{z}\leq p_{z}^{\prime }\right\}} +q_{z}\bm1_{\left\{
		p_{z}>p_{z}^{\prime }\right\}} =q_{z}$, and hence requirement~(\ref{req2}) is met.
	
	\item[(b)] $\max p_{z}\leq \max p_{z}^{\prime }$. Then set $q^{\vee }=q^{\prime }$%
	, and notice that $q_{z}^{\vee }=q_{z}^{\prime }>0$ implies $p_{z}\leq p_{z}^{\prime }$,
	giving $q_{z}^{\prime }\bm1_{\left\{ p_{z}\leq p_{z}^{\prime }\right\}}
	+q_{z}\bm1_{\left\{ p_{z}>p_{z}^{\prime }\right\}} =q_{z}^{\prime }$, and hence
	requirement~(\ref{req2}) is met.
	
\end{itemize}

\noindent [{\sc Step} (ii)]  In order to show (ii), we first show that for $y$ and $z$ in $\{1,\ldots,N\}$, we have
\begin{equation}
	\left(q_{z}       \bm1_{\left\{ p_{z}\leq p_{z}^{\prime }\right\}} \right) 
	\left(q_{y}^{\prime }\bm1_{\left\{ p_{y}>p_{y}^{\prime }\right\}} \right) =0.
	\label{complementarity}
\end{equation}
Indeed, assume otherwise. The definition of $\mathtt Q$ implies that $q_{z}>0$
implies $p_{z}\geq p_{y}$, and $q_{y}^{\prime }>0$ implies $p_{y}^{\prime
}\geq p_{z}^{\prime }$, which we combine with the negation of~(\ref%
{complementarity}) to obtain
\begin{equation*}
	p_{z}^{\prime }\geq p_{z}\geq p_{y}>p_{y}^{\prime }\geq p_{z}^{\prime },
\end{equation*}%
a contradiction. 

Then define
\begin{eqnarray*}
	\mathcal{M} &=&\left\{ z:q_{z}\bm1_{\left\{ p_{z}\leq
		p_{z}^{\prime }\right\}} >0\right\}  \\
	\mathcal{M}^{\prime } &=&\left\{ z:q_{z}^{\prime }\bm1_{\left\{
		p_{z}>p_{z}^{\prime }\right\}} >0\right\}.
\end{eqnarray*}%
Note that it follows from~(\ref{complementarity}) that if $\mathcal{M}$ is
nonempty then $\mathcal{M}^{\prime }$ is empty and conversely. We thus
consider three cases:

\begin{itemize}
	
	\item[(a)] Assume that $\mathcal{M}$ is nonempty. Then $q_{y}^{\prime }\bm1_{\left\{
		p_{y}>p_{y}^{\prime }\right\}} =0$ for all $y\in \{1,\ldots,N\}$. Note that for $%
	z\in \mathcal{M}$ and for $y\in \{1,\ldots,N\}$, $p_{z}\geq p_{y}$, and
	therefore $p_{z}=p_{z}\wedge p_{z}^{\prime }\geq p_{y}\wedge p_{y}^{\prime }$%
	, therefore $z\in \mathcal{M}\implies z\in \arg \max_{y\in \mathcal{Z}%
	}\left\{ p_{y}\wedge p_{y}^{\prime }\right\} $. Thus, set
	\begin{equation*}
		q_{z}^{\wedge }=\frac{q_{z}\bm1_{\left\{ p_{z}\leq p_{z}^{\prime }\right\}} }{%
			\sum_{y\in \mathcal{M}}q_{y}\bm1_{\left\{ p_{y}\leq p_{y}^{\prime }\right\}} }
	\end{equation*}%
	and~(\ref{req1}) is met.
	
	\item[(b)] Assume that $\mathcal{M}^{\prime }$ is nonempty. Then $q_{z}\bm1_{\left\{
		p_{z}\leq p_{z}^{\prime }\right\}} =0$ for all $y\in \{1,\ldots,N\}$. Note that
	for $z\in \mathcal{M}^{\prime }$ and for $y\in \{1,\ldots,N\}$, $p_{z}^{\prime
	}\geq p_{y}^{\prime }$, and therefore $p_{z}^{\prime }=p_{z}\wedge
	p_{z}^{\prime }\geq p_{y}\wedge p_{y}^{\prime }$, therefore $z\in \mathcal{M}%
	^{\prime }\implies z\in \arg \max_{y\in \mathcal{Z}}\left\{ p_{y}\wedge
	p_{y}^{\prime }\right\} $. Thus, set
	\begin{equation*}
		q_{z}^{\wedge }=\frac{q_{z}^{\prime }\bm1_{\left\{ p_{z}>p_{z}^{\prime }\right\}}
		}{\sum_{y\in \mathcal{M}^{\prime }}q_{y}^{\prime }\bm1_{\left\{
				p_{y}>p_{y}^{\prime }\right\}} }
	\end{equation*}%
	and (\ref{req1}) is met.
	
	\item[(c)] If both $\mathcal{M}$ and $\mathcal{M}^{\prime }$ are empty, then~(\ref%
	{req1}) is trivially met by any $q^{\wedge }$ supported in $\arg \max
	\left\{ p\wedge p^{\prime }\right\} $.
\end{itemize}

\noindent  This completes the argument.

\subsection{Nonreversingness and Aggregate Monotonicity but Not Monotone Total Output}\label{kettle}

We give an example of a point-valued correspondence that satisfies aggregate and weighted monotonicity and hence nonreversingness but
not the monotone total output condition defined in Property \ref{zeus}.  
Consider the function
\[
\mathtt q(p) = \arg\max_{q\in\mathbb R^3} p^{\top}q-q^{\top}Cq,
\]
where the matrix $C$ is given by
\[
C=\left[
\begin{array}{ccc}
	25&10&24\\
	10&5&10\\
	24&10&25
\end{array}
\right].
\]
We can solve for the supply function $\mathtt q(p) =\frac12 C^{-1}p$, 
where
\[
C^{-1} = \frac{1}{90}
\left[
\begin{array}{ccc}
	~~50&-20&-40\\
	-20&~~98&-20\\
	-40&-20&~~50
\end{array}
\right]
\]
and the inverse function $\mathtt q^{-1}(q)$  is given by $\mathtt q^{-1}(q) = 2Cq$.

The function $\mathtt q$ satisfies unified gross substitutes (it suffices that $C^{-1}$ has positive entries on the diagonal and negative entries elsewhere)  and the inverse correspondence $\mathtt q^{-1}$ is isotone (because $C$ is positive), and hence $\mathtt q(p)$ is nonreversing (cf. Property \ref{monica}).  However, monotone total output fails.  For example, $\mathtt q(0,0,0) = (0,0,0)$ while $\mathtt q(2,0,0) = (50,-20,-40)$.  Aggregate monotonicity is satisfied as is weighted monotonicity (for example, set  $k = (2,1,1)$ for this pair of prices). 
\hfill\rule{.1in}{.1in}


		%
		%

\subsection{Equivalent Definitions of M-functions}\label{water}

We first show that our Definition \ref{shade} implies More and Rheinboldt's \cite[Definition 2.3, p. 48]{more1973p} definition of an M-function.  Hence, let $\mathtt Q$ be a function satisfying unified gross substitutes and nonreversingness, with a point-valued inverse.  Property \ref{prop:wgs-implies-ugs} establishes that $\mathtt Q$ satisfies weak gross substitutes (and hence in More and Rheinboldt's \cite{more1973p}
language is off-diagonally antitone). By Theorem \ref{belgian}, the inverse function $\mathtt Q^{-1}$ is isotone in the strong set order, and hence for $q\leq q^{\prime }$ with $p\in \mathtt Q^{-1}\left(
q\right) $ and $p^{\prime }\in \mathtt Q^{-1}\left( q^{\prime }\right) $, we have $p\wedge p^{\prime }\in \mathtt Q^{-1}\left( q\right) $ and $p\vee p^{\prime }\in
\mathtt Q^{-1}\left( q^{\prime }\right) $. However, as $\mathtt Q^{-1}$ is point valued, $%
Q^{-1}\left( q\right) =\left\{ p\right\} $ and $Q^{-1}\left( q^{\prime
}\right) =\left\{ p^{\prime }\right\} $ and thus $p\wedge p^{\prime }=p$ and 
$p\vee p^{\prime }=p^{\prime }$.  Hence $\mathtt Q(p)\le \mathtt Q(p')\implies p\le p'$, and so $\mathtt Q$  is an M-function in the sense of More-Rheinboldt.

Conversely, assume that $\mathtt Q$ is a M-function in the sense of More-Rheinboldt.
Then by Property \ref{prop:wgs-implies-ugs}, the point-valued correspondence $\mathtt  Q $
satisfies unified gross substitutes.   Because we have (for $q\in \mathtt Q(p)$ and $q'\in \mathtt Q(p')$) 
\begin{equation}\label{hockey}
	q\le q' \implies p\le p',    
\end{equation} 
the inverse $\mathtt Q^{-1}$ is point valued.  (In particular, applying this implication and $q\le q$ to $q\in \mathtt Q(p)$ and $q\in \mathtt Q(p')$ gives $p\le p'$ and $p'\le p$).  Finally, \eqref{hockey} implies that the antecedent of nonreversingness can hold only if $p=p'$, at which point the consequent must hold, giving nonreversingness and hence Definition \ref{shade}.\hfill\rule{.1in}{.1in}

\subsection{Imperfect Competition}\label{sawbones}
In this subsection, we assume throughout that $Q$ is a compact subset of $\mathbb R ^ {\mathcal Z}$. We consider the supply correspondence given by 
$$
\mathtt Q\left( p\right) =\arg \max_{q \in Q}\left\{ \pi \left( p,q\right) -c\left(
q\right) \right\}
$$
which is associated with indirect profit function 
$$
c^{\pi}\left( p\right) =\max_{q \in Q}\left\{ \pi \left( p,q\right) -c\left(
q\right) \right\},
$$
where we assume that the profit function is such that
$\pi \left( p,q\right) =\sum_{z}\pi _{z}\left( p_{z},q_{z}\right) $ for some functions $\pi_z(.,.)$.

We investigate which restrictions on $\pi_z$ are needed to maintain necessary and sufficient conditions that relate unified gross substitutes of $\mathtt Q$ to the submodularity of the indirect profit function $c^{\pi}$.

\begin{proposition}\label{prop:ugs-implies-submod-gal}
	Assume $\pi \left( p,q\right) =\sum_{z}\pi _{z}\left( p_{z},q_{z}\right) $, 
	where $\partial _{pq}^{2}\pi _{z}\left( p_{z},q_{z}\right) $ exists and is
	nonnegative and is continuous, and $\pi _{z}\left( p_{z},q_{z}\right) $ is
	concave in $q_{z}$. Then unified gross substitutes of $\partial c^{\pi}\left( p\right) $
	implies submodularity of $c^{\pi}$.
\end{proposition}

\begin{proposition}\label{prop:submod-implies-ugs-gal}
	Assume $\pi \left( p,q\right) =\sum_{z}\pi _{z}\left( p_{z},q_{z}\right) $, $%
	\partial _{pq}^{2}\pi _{z}\left( p_{z},q_{z}\right) $ exists and is positive
	and is continuous, and both $\pi _{z}\left( p_{z},q_{z}\right) $ and $%
	\partial _{p_{z}}\pi _{z}\left( p_{z},q_{z}\right) $ are concave in $q_{z}$.
	Then submodularity of $c^{\pi}$ implies unified gross substitutes of $\partial c^{\pi}\left(
	p\right) $.
\end{proposition}

The proofs of these two propositions rely on the following lemma, which is part of a family of  envelope theorems generically known as Danskin's theorem, which we state here without a proof (for the proof we refer to Bonnans and Shapiro~\cite[Theorem 4.16]{bonnans2013perturbation}:

\begin{lemma}
	\label{lemma-1}Assume $\pi \left( p,q\right) $ is concave in $q$, $\partial _{p}\pi \left(
	p,q\right) $ exists, is nonnegative and is continuous in $\left( p,q\right) $%
	. Given a function $c$ continuous on $Q$, one has%
	\[
	\frac{dc^{\pi}\left( p+tb\right) }{dt}|_{0^{+}}=\max_{q \in Q}\left\{ b^{\top
	}\partial _{p}\pi \left( p,q\right) :q\in {\mathtt Q} \left( p\right) \right\} . 
	\]
\end{lemma}

\paragraph{Proof of Proposition~\ref{prop:ugs-implies-submod-gal}.}
Assume unified gross substitutes holds for the correspondence $p\rightrightarrows \partial c^{\pi}\left( p\right) $. Take $b\geq 0$, and define $b^{\leq}$ and $b^{>}$ as in\eqref{eq:border}. We want to show (\ref{shuttle}), which, by virtue of Lemma~\ref{lemma:diff-char-submod}, will allows us to
conclude that $c^{\pi}$ is submodular. Then
consider $q\in {\mathtt Q} \left( p\right) $ that attains $\max_{q}\left\{ \left(
b^{\leq }\right) ^{\top }\partial _{p}\pi \left( p,q\right) :q\in {\mathtt Q} \left(
p\right) \right\} $, and take $q^{\prime }\in {\mathtt Q} \left( p^{\prime }\right) $
that attains $\max_{q^{\prime } \in Q}\left\{ \left( b^{>}\right) ^{\top }\partial
_{p}\pi \left( p^{\prime },q^{\prime }\right) :q^{\prime }\in {\mathtt Q} \left(
p^{\prime }\right) \right\} $. By unified gross substitutes, there exists $q^{\wedge }\in {\mathtt Q} \left(
p\wedge p^{\prime }\right) \,$such that%
\[
q\mathbf1_{\mathcal{Z}^{\leq}}+q^{\prime }\mathbf1_{\mathcal{Z}^{>}}\leq q^{\wedge }.
\]
By increasing differences and by unified gross substitutes, we get that 
\begin{eqnarray*}
	&&\max_{q}\left\{ \left( b^{\leq }\right) ^{\top }\partial _{p}\pi \left(
	p,q\right) :q\in {\mathtt Q} \left( p\right) \right\} +\max_{q^{\prime }}\left\{ \left(
	b^{>}\right) ^{\top }\partial _{p}\pi \left( p^{\prime },q^{\prime }\right)
	:q^{\prime }\in {\mathtt Q} \left( p^{\prime }\right) \right\}  \\
	&=&\left( b^{\leq }\right) ^{\top }\partial _{p}\pi \left( p\wedge p^{\prime
	},q\right) +\left( b^{>}\right) ^{\top }\partial _{p}\pi \left( p\wedge
	p^{\prime },q^{\prime }\right)  \\
	&=&\partial _{p}\pi \left( p\wedge p^{\prime },\left( b^{\leq }\right)
	^{\top }q+\left( b^{>}\right) ^{\top }q^{\prime }\right) \leq \partial
	_{p}\pi \left( p\wedge p^{\prime },q^{\wedge }\right)  \\
	&\leq &\max_{\tilde{q}}\left\{ b^{\top }\partial _{p}\pi \left( p\wedge
	p^{\prime },\tilde{q}\right) :\tilde{q}\in {\mathtt Q} \left( p\wedge p^{\prime
	}\right) \right\} 
\end{eqnarray*}%
and as a result of Lemma~\ref{lemma-1}, we get that 
\[
\frac{dc^{\pi}\left( p+tb^{\leq }\right) }{dt}|_{0^{+}}+\frac{dc^{\ast
	}\left( p+tb^{>}\right) }{dt}|_{0^{+}}=\left( b^{\leq }\right) ^{\top
}\partial _{p}\pi \left( p,q\right) +\left( b^{>}\right) ^{\top }\partial
_{p}\pi \left( p^{\prime },q^{\prime }\right) .
\]%
\hfill\rule{.1in}{.1in}

\paragraph{Proof of Proposition~\ref{prop:submod-implies-ugs-gal}.}
To show the converse, assume $c^{\pi}$ is submodular, or equivalently (by
successively applying Lemma~\ref{lemma:diff-char-submod} and Lemma~\ref{lemma-1}), that 
\[
\sup_{q\in \partial c^{\pi}\left( p\right) }\left( b^{\leq }\right) ^{\top
}\partial _{p}\pi \left( p,q\right) +\sup_{q^{\prime }\in \partial c^{\ast
	}\left( p^{\prime }\right) }\left( b^{>}\right) ^{\top }\partial _{p}\pi
\left( p^{\prime },q^{\prime }\right) \leq \sup_{q^{\wedge }\in \partial
	c^{\pi}\left( p\wedge p^{\prime }\right) }b^{\top }\partial _{p}\pi \left(
p\wedge p^{\prime },q^{\wedge }\right) 
\]%
holds for all $b\geq 0$. Take $q\in \partial c^{\pi}\left( p\right) $ and $%
q\in \partial c^{\pi}\left( p^{\prime }\right) $. We have for all $b\geq 0$%
, 
\[
b^{\top }\left( q\mathbf1_{\mathcal{Z}^{\leq}}+q^{\prime }\mathbf1_{\mathcal{Z}^{>}}\right) =qb^{\leq }+q^{\prime
}b^{>},
\]%
but%
\begin{eqnarray*}
	\left( b^{\leq }\right) ^{\top }\partial _{p}\pi \left( p,q\right)  &\leq
	&\sup_{q\in \partial c^{\pi}\left( p\right) }\left( b^{\leq }\right)
	^{\top }\partial _{p}\pi \left( p,q\right)  \\
	\left( b^{>}\right) ^{\top }\partial _{p}\pi \left( p^{\prime },q^{\prime
	}\right)  &\leq &\sup_{q^{\prime }\in \partial c^{\pi}\left( p^{\prime
		}\right) }\left( b^{>}\right) ^{\top }\partial _{p}\pi \left( p^{\prime
	},q^{\prime }\right) 
\end{eqnarray*}%
and therefore 
\[
b^{\top }\partial _{p}\pi \left( p\wedge p^{\prime },q\mathbf1_{\mathcal{Z}^{\leq}}+q^{\prime
}\mathbf1_{\mathcal{Z}^{>}}\right) \leq \sup_{q^{\wedge }\in \partial c^{\pi}\left( p\wedge
	p^{\prime }\right) }b^{\top }\partial _{p}\pi \left( p\wedge p^{\prime
},q^{\wedge }\right) 
\]%
holds for all $b\geq 0$. By Lemma~\ref{lem:convex-envelope-dominating}, this implies that there is an element in the closed convex hull of $ \partial c^{\pi}\left( p\wedge
p^{\prime }\right) $ which is greater or equal than $\partial _{p}\pi \left( p\wedge p^{\prime },q\mathbf1_{\mathcal{Z}^{\leq}}+q^{\prime
}\mathbf1_{\mathcal{Z}^{>}}\right)$ in the partial order. This element can be represented as an integral over elements in $ \partial c^{\pi}\left( p\wedge
p^{\prime }\right) $, and therefore it can be written as $\int_{0}^{1}\partial _{p}\pi \left( p\wedge p^{\prime
},q_{t}^{\wedge }\right) d\mu \left( t\right) 
$ . That is,
\[
\partial _{p}\pi \left( p\wedge p^{\prime },q\mathbf1_{\mathcal{Z}^{\leq}}+q^{\prime
}\mathbf1_{\mathcal{Z}^{>}}\right) \leq \int_{0}^{1}\partial _{p}\pi \left( p\wedge p^{\prime
},q_{t}^{\wedge }\right) d\mu \left( t\right) 
\]%
where $q_{t}^{\wedge }\in \partial c^{\pi}\left( p\wedge p^{\prime
}\right) $. Because $\partial _{p}\pi \left( p,q\right) $ is concave in $q$,
we get 
\[
\int_{0}^{1}\partial _{p}\pi \left( p\wedge p^{\prime },q_{t}^{\wedge
}\right) d\mu \left( t\right) \leq \partial _{p}\pi \left( p\wedge p^{\prime
},\int_{0}^{1}q_{t}^{\wedge }d\mu \left( t\right) \right) 
\]%
and denoting $q^{\wedge }=\int_{0}^{1}q_{t}^{\wedge }d\mu \left( t\right) $,
we have $q^{\wedge }\in \partial c^{\pi}\left( p\wedge p^{\prime }\right) $
by the convexity of $\partial c^{\pi}\left( p\wedge p^{\prime }\right) $,
and therefore 
\[
\partial _{p}\pi \left( p\wedge p^{\prime },q\mathbf1_{\mathcal{Z}^{\leq}}+q^{\prime
}\mathbf1_{\mathcal{Z}^{>}}\right) \leq \partial _{p}\pi \left( p\wedge p^{\prime },q^{\wedge
}\right) ,
\]%
and by the fact that $\partial _{p}\pi \left( p,q\right) $ is increasing in $%
q$, we get that 
$q\mathbf1_{\mathcal{Z}^{\leq}}+q^{\prime }\mathbf1_{\mathcal{Z}^{>}}\leq q^{\wedge }$ as required.
\hfill\rule{.1in}{.1in}

\subsection{Partial Inverse Correspondence}
Let $X^c=\{1,\ldots,N\}\setminus X$.  We can generalize Theorem \ref{belgian}.  Define the \emph{partial inverse correspondence} at set $X\subseteq \{1,\ldots,N\}$ as
\[
\mathtt Q_{X}^{-1}\left( q_{X};p_{X^c}\right) =\left\{
p_{X}:\exists q_{X^c},\left( q_{X},q_{X^c}\right)
\in \mathtt Q\left( p_{X},p_{X^c}\right) \right\} .
\]
Note that for  $X = \mathcal \{1,\ldots,N\}$, this is the inverse correspondence $\mathtt Q^{-1}$.

\begin{corollary}
	If $\mathtt Q$ is an M0-correspondence, then the partial inverse correspondence $\ \mathtt Q_{X}^{-1} $ is isotone in the strong set order for any $X\subseteq \{1,\ldots,N\}$.
\end{corollary}

\paragraph{Proof.} Assume $\mathtt Q$ is an M0-correspondence and take $p_{X}\in \mathtt Q_{X}^{-1}\left( q_{X};p_{X^c}\right) $ and $p_{X}^{\prime }\in \mathtt Q_{X}^{-1}\left(
q_{X}^{\prime };p_{X^c}
\right) $. 
Let $q_X\le q'_X$.  
We need to show that $\left( p\wedge p^{\prime
}\right) _{X}\in \mathtt Q_{X}^{-1}\left( q_{X},p_{X^c}\right) $
and $\left( p\vee p^{\prime }\right) _{X}\in \mathtt Q_{X}^{-1}\left( q_{X}^{\prime
},p_{X^c}^{\prime }\right) $. By unified gross substitutes, there exist 
$q^{\wedge }\in \mathtt Q\left( p\wedge p^{\prime }\right) $ and $q^{\vee }\in
\mathtt Q\left( p\vee p^{\prime }\right) $ such that $p_{z}\leq p_{z}^{\prime
}\implies q_{z}\leq q_{z}^{\wedge }$ and $p_{z}>p_{z}^{\prime }\implies
q_{z}^{\prime }\leq q_{z}^{\wedge }$. But $p_{z}>p_{z}^{\prime }\implies
z\in X\implies q_{z}\leq q_{z}^{\prime }\implies q_{z}\leq q_{z}^{\wedge }$,
and therefore $q\leq q^{\wedge }$. As a result of the fact that $\mathtt Q$
is nonreversing, we have $q\in \mathtt Q\left( p\wedge p^{\prime }\right) $, and
thus $\left( p\wedge p^{\prime }\right) _{X}\in \mathtt Q_{X}^{-1}\left( q_{X},p_{X^c}\right) $. A similar logic shows that $\left( p\vee
p^{\prime }\right) _{X}\in \mathtt Q_{X}^{-1}\left( q_{X}^{\prime },p_{X^c}^{\prime }\right) $. \hfill\rule{.1in}{.1in}

\subsection{Structures of Solutions}\label{wolverine}

\begin{definition}[Subsolutions,  supersolutions, and solutions]
	A price vector $p\in P$ is a \emph{subsolution} of $\mathtt Q $ if there exists $q\in \mathtt Q\left( p\right) $ with $q\leq 0$.
	It is a \emph{supersolution} of $\mathtt Q $ if there exists $q\in \mathtt Q\left( p\right) $ with $q\geq 0$.
	It is a \emph{solution} of $\mathtt Q $ if there exists $q\in \mathtt Q\left( p\right) $ with $q = 0$.
\end{definition}

\noindent Note that a solution of $\mathtt Q$ is both a subsolution and a supersolution. However, the converse is not true in general: if $p$ is both a subsolution and a supersolution, it need not be a solution without additional assumptions.

\begin{proposition}
	Assume  $\mathtt Q$ satisfies unified gross substitutes. Then set of subsolutions of $\mathtt Q $ is closed under $\vee $, and the set of
	supersolutions is closed under $\wedge $.
\end{proposition}

\paragraph{Proof}
Let $p$ and $p^{\prime }$ be two subsolutions. Then there is $q\in Q\left(
p\right) $ with $q\leq 0$ and $q^{\prime }\in Q\left( p^{\prime }\right) $
with $q^{\prime }\leq 0$.
By unified gross substitutes, there is $q^{\vee }\in Q\left( p\vee p^{\prime }\right) $ such that $%
p_{z}\leq p_{z}^{\prime }\implies q^{\vee }\leq q_{z}^{\prime }\leq 0$, and $%
p_{z}>p_{z}^{\prime }\implies q^{\vee }\leq q_{z}\leq 0$. Thus $q^{\vee
}\leq 0$. The result for supersolutions is treated similarly.
\hfill\rule{.1in}{.1in}\bigskip

\noindent As a consequence of the proposition, assuming only unified gross substitutes,  if we take $p$ and $p^\prime$ to be two solutions of $\mathtt Q$, then $p \vee p^\prime$ is a supersolution, and $p \wedge p^\prime$ is a subsolution. However, there is no guarantee that they will be solutions. If we assume further that $\mathtt Q$ is nonreversing, and hence that it is a M0-correspondence, then we get that $p \vee p^\prime$ and $p \wedge p^\prime$ are solutions by virtue of Corollary~\ref{gobble}.

\begin{proposition}\label{prop:equiv-sol-max-subsol}
	Let $\mathtt Q$ be a M0-correspondence.  If a solution of $\mathtt Q$ exists, and if a maximal subsolution exists, then these two elements coincide.
\end{proposition}

\paragraph{Proof}
Let  $p^{\ast }$ be a solution of $\mathtt Q$ and hence $0\in Q\left( p^{\ast }\right) $. Let $\bar{p}$ be a maximal subsolution. Then $p^{\ast }$ is a subsolution
and $p^{\ast }\leq \bar{p}$ holds by definition of $\bar{p}$. As the latter vector is a subsolution, there is $q\in Q\left(
\bar{p}\right) $ with $0\geq q$. \ By nonreversingness it follows that $0\in
Q\left( p^{\ast }\right) $.
\hfill\rule{.1in}{.1in}\bigskip

\noindent Given that the set of subsolutions is closed under $\vee$, this set is a semi-sublattice of $\mathbb R^N$. For a maximal element to exist, one requires this semi-sublattice to be complete. This will hold in particular if $P$ is finite, but also more generally under continuity assumptions on $\mathtt Q$.

\end{document}